\documentclass[final]{arxiv}

\usepackage{bm}
\usepackage{mathrsfs}
\usepackage{mathtools}
\usepackage{graphicx}
\usepackage{subcaption}
\usepackage[export]{adjustbox}
\usepackage{float}
\usepackage{hyperref}
\hypersetup{colorlinks=true, linkcolor=blue, citecolor=blue, filecolor=blue, urlcolor=blue}

\newcommand{\E}{\mathbb{E}}
\renewcommand{\P}{\mathbb{P}}
\newcommand{\R}{\mathbb{R}}

\newcommand{\cB}{\mathcal{B}}
\newcommand{\cE}{\mathcal{E}}
\newcommand{\cF}{\mathcal{F}}
\newcommand{\cK}{\mathcal{K}}
\newcommand{\cM}{\mathcal{M}}
\newcommand{\cP}{\mathcal{P}}
\newcommand{\cS}{\mathcal{S}}
\newcommand{\cW}{\mathcal{W}}
\newcommand{\cX}{\mathcal{X}}

\newcommand{\one}{\bm 1}

\newcommand{\bH}{\bm{H}}
\newcommand{\bZ}{\bm{Z}}

\newcommand{\bz}{\bm{z}}

\newcommand{\sX}{\mathscr{X}}
\newcommand{\sY}{\mathscr{Y}}

\newcommand{\sfH}{\mathsf{H}}
\newcommand{\sfK}{\mathsf{K}}

\renewcommand{\d}{\mathrm{d}}
\newcommand{\MMD}{\mathrm{MMD}}

\newcommand{\ra}{\rightarrow}

\newcommand{\dto}{\overset{D}{\rightarrow}}
\newcommand{\bsfK}{\bar\sfK}
\newcommand{\vep}{\varepsilon}
\newcommand{\Var}{\mathrm{Var}}

\newtheorem{assumption}{Assumption}[section]

\numberwithin{equation}{section}

\title[A Martingale Kernel Two-Sample Test]{A Martingale Kernel Two-Sample Test}
\usepackage{times}

\altauthor{%
 \Name{Anirban Chatterjee} \Email{anirbanc@uchicago.edu}\\
 \addr Department of Statistics, University of Chicago
 \AND
 \Name{Aaditya Ramdas} \Email{aramdas@stat.cmu.edu}\\
 \addr Department of Statistics and Data Science, Machine Learning Department, Carnegie Mellon University
}

\begin{document}

\maketitle

\begin{abstract}%
  The Maximum Mean Discrepancy (MMD) is a widely used multivariate distance metric for two-sample testing. The standard MMD test statistic has an intractable null distribution typically requiring costly resampling or permutation approaches for calibration. In this work we leverage a martingale interpretation of the estimated squared MMD to propose martingale MMD (mMMD), a quadratic-time statistic which has a limiting standard Gaussian distribution under the null. Moreover we show that the test is consistent against any fixed alternative and for large sample sizes, mMMD offers substantial computational savings over the standard MMD test, with only a minor loss in power.
\end{abstract}



\begingroup
\setcounter{tocdepth}{2} 
\endgroup

\section{Introduction}

Given two distributions $P$ and $Q$ over a metric space $\cX$, the two-sample problem is to test
\begin{align}\label{eq:H01}
\bH_0: P = Q\text{ versus }\bH_1: P \neq Q
\end{align}
based on independent samples $\sX_n = \{X_1,\ldots, X_n\} \sim P$ and $\sY_n = \{Y_1,\ldots, Y_n\} \sim Q$. This classical problem has been widely studied in the parametric setting, assuming low-dimensional structure. However, parametric methods often break down under model misspecification or non-Euclidean data, motivating flexible non-parametric approaches with minimal assumptions. For univariate data, classical non-parametric tests include the Kolmogorov–Smirnov test \citep{smirnov1948table}, Wald–Wolfowitz runs test \citep{wald1940test}, rank-sum test \citep{mann1947test, wilcoxon1945individual}, and Cramér–von Mises test \citep{anderson1962distribution}. Early multivariate extensions by \citet{weiss1960two} and \citet{bickel1969distribution} inspired modern methods based on geometric graphs \citep{friedman1979multivariate, rosenbaum2005exact, bhattacharya2019general}, energy distance \citep{szekely2003statistics, szekely2013energy}, kernel MMD \citep{gretton2009fast, gretton2012kernel, ramdas2015decreasing, shekhar2022permutation, chatterjee2025boosting}, ball divergence \citep{pan2018ball, banerjee2025high}, classifier-based tests \citep{lopez2017revisiting, kim2021classification}, and others.

Among the tests mentioned, kernel-based methods have recently emerged as powerful tools for detecting distributional differences in general domains. These methods use the Maximum Mean Discrepancy (MMD) (see \eqref{eq:def_MMD}) to compare distributions $P$ and $Q$ by embedding them into a reproducing kernel Hilbert space (RKHS) $\cK$ associated with a positive definite kernel $\sfK$. \citet{gretton2012kernel} proposed a quadratic-time estimator of the squared MMD, forming the basis of the kernel two-sample test (referred to as the quadratic-time $\MMD$ test) which rejects the null $\bH_0$ when the statistic exceeds a threshold $\tau_\alpha$ that controls the Type I error at level $\alpha$.

Despite its popularity, a key practical challenge in the $\MMD$ test is determining the threshold $\tau_\alpha$. Under $\bH_0$, the test statistic weakly converges to an infinite weighted sum of centered $\chi^2$ variables, with weights depending on the unknown distribution $P = Q$~\citep{gretton2009fast, gretton2012kernel}. As a result, estimating $\tau_\alpha$ typically requires either computationally intensive resampling \citep{gretton2009fast, chatterjee2025boosting} or moment-based parametric approximations (for example using Gamma, Pearson curves) \citep{gretton2012kernel}, which often lack consistency or accuracy guarantees. Calibration via concentration bounds is another option but tends to be overly conservative. Recent work has proposed more efficient alternatives \citep[Section~6]{gretton2012kernel}; \citep{zaremba2013b, jitkrittum2017linear}, but these often overlook the pairwise kernel structure and can suffer from low power. We review these methods in detail in Section~\ref{sec:related_works}.

In this work, we propose a simple and novel variant of the kernel-based two-sample test, where the test statistic follows an asymptotic standard normal distribution under $\bH_0$, making it easy to implement in practice. It remains valid under both fixed kernels (independent of sample size) and varying kernels (dependent on sample size) settings. We show that the test is consistent and  establish asymptotic normality under fixed alternatives as well. Below, we provide a brief overview of our result, with full details deferred to Sections~\ref{sec:derive_mMMD_test} and~\ref{sec:null_dist_section}.

\subsection{Overview of the Proposed Method}

\begin{figure}[!h]
    \centering
    \subfigure[]{\includegraphics[width=0.3\textwidth]{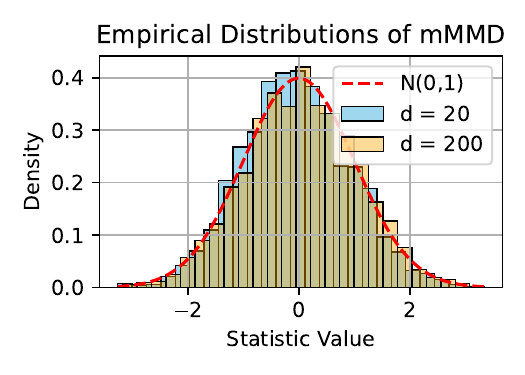}} 
    \subfigure[]{\includegraphics[width=0.3\textwidth]{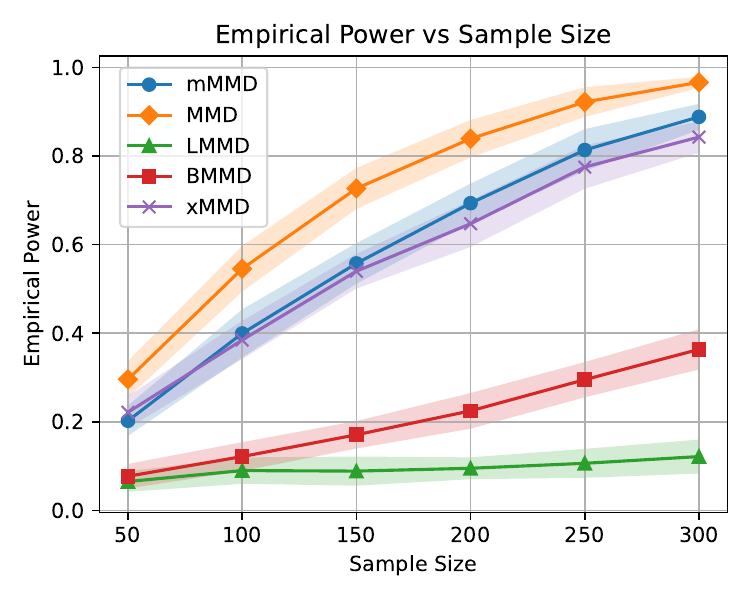}} 
    \subfigure[]{\includegraphics[width=0.37\textwidth]{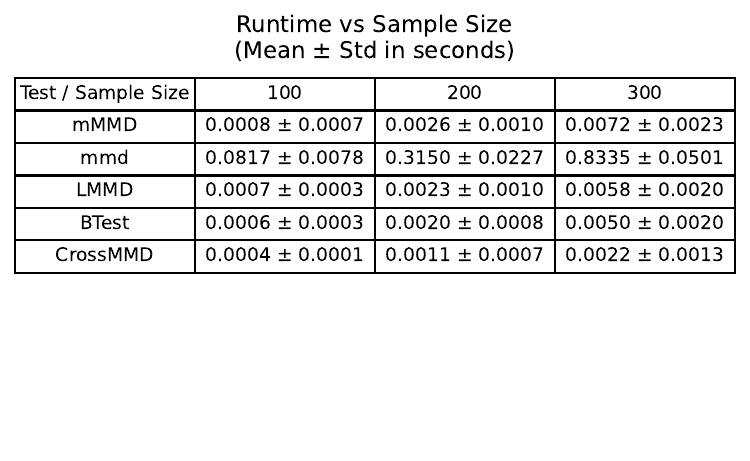}}
    \caption{The $m\MMD$-test: (a) Empirical distribution of the $m\MMD$-test statistic (see \eqref{eq:def_eta_n} for a formal definition) under $\bH_0$ with $d = 20,200$ and $P = Q$ are $d$-dimensional standard Gaussian distributions. (b) The second figure compares power of the $m\MMD$ test with the quadratic time MMD test (implemented with $200$ permutations) from \cite{gretton2012kernel} and computationally efficient variants from \cite{gretton2012kernel}, \cite{zaremba2013b} and \cite{shekhar2022permutation}. 
    The main takeaway is that the MMD is the most powerful but most computationally inefficient. The $x\MMD$ and $m\MMD$ are more powerful than the others, but the latter avoids sample splitting.
    (c) The third figure shows the computational efficiency of the proposed test against the permutation-based quadratic time MMD test. All tests use the Gaussian kernel with median bandwidth.}
    \label{fig:overview}
\end{figure}
\noindent
We propose a variant of the quadratic-time $\MMD$-test (see \cite{gretton2012kernel}) motivated by the key observation that a modified form of the estimator of the squared $\MMD$ admits a martingale structure. Specifically, we introduce the \emph{martingale-$\MMD$} $(m\MMD)$ statistic, defined as
\begin{align*}
    T_n = \frac{1}{n}\sum_{i=2}^{n}\left\langle \hat f_i, \sfK(X_i,\cdot) - \sfK(Y_i, \cdot)\right\rangle_{\cK}, \quad \text{where} \quad \hat f_i = \frac{1}{i}\sum_{j=1}^{i-1}\sfK(X_j,\cdot) - \sfK(Y_j,\cdot).
\end{align*}
In other words, for each $2 \leq i \leq n$, our method estimates the \emph{witness function} (see Section~\ref{sec:derive_mMMD_test} and \cite{kubler2022witness}) using the empirical mean embeddings computed from past data, and then averages the resulting estimates of the squared $\MMD$. The name arises because $nT_n$ is a (zero-mean) martingale under the null.

Our final test statistic (for \eqref{eq:H01}) is then defined as $\eta_n = T_n/\sigma_n$ where $\sigma_n$ (see \eqref{eq:sigma_n_exp}) is an estimate of the asymptotic variance under $\bH_0$. In Section~\ref{sec:null_dist_section}, we establish general results showing that $\eta_n$ converges in distribution to a standard Gaussian under the null. Theorem~\ref{thm:null_clt_fixed} shows this under mild moment conditions, when $\sfK, P$ and the data distribution remains fixed with respect to $n$. Theorem~\ref{thm:null_clt_general} then established conditions under which a Berry--Esseen-type convergence to $\mathrm N(0,1)$ holds when the kernel and both data distribution and data dimension are allowed to change with $n$.

Leveraging this asymptotic Gaussian behavior, our main methodological contribution is the \emph{mMMD-test}, which rejects the null hypothesis if
$\eta_n > z_{1 - \alpha}$, where $z_{1 - \alpha}$ is the $(1 - \alpha)$ quantile of the standard normal distribution. Formally, we define the test as
\begin{align}\label{eq:mMMD_test}
    \Psi_n := \mathbf{1} \left\{ \eta_n > z_{1 - \alpha} \right\}.
\end{align}

\noindent
Turning to performance under alternatives, in Theorem~\ref{thm:gen_consistency} we identify general conditions under which the test remains uniformly consistent over a class of alternatives, even in the challenging setting where $\mathsf{K},P, Q$ and the data dimension may vary with the sample size $n$. Furthermore, Theorem~\ref{thm:general_alt_dist} shows that under the alternative hypothesis $P_n \neq Q_n$, the test statistic (properly centered and scaled) converges in distribution to a standard Gaussian. This result directly implies (see Theorem~\ref{thm:fixed_consistency}) that the test defined in \eqref{eq:mMMD_test} is consistent in the simpler case when the parameters do not depend on $n$.

In Section~\ref{sec:experiments}, we empirically validate our method and compare it against the quadratic-time $\MMD$ test and other computationally efficient variants, demonstrating strong performance on both simulated and real data. Section~\ref{sec:broad_scope} explores extensions to multiple-kernel settings and introduces a broader class of MMD-based statistics built on the same principles as the $m\MMD$. Theorem~\ref{thm:berry_essen_H0} shows that, under suitable conditions, this class also admits an asymptotic standard Gaussian distribution. Furthermore, in Section~\ref{sec:minimax}, we show that certain statistics in this class remain consistent even when the $L_2$ distance between the densities decays at the minimax rate.

\subsection{Related Works}\label{sec:related_works}
The literature on kernel-based two-sample tests is quite extensive. In this section, we succinctly review the most relevant works. The first approach to obtaining a distribution-free kernel two-sample test relied on concentration-inequality-based large deviation bounds for the quadratic-time MMD statistic \citep{gretton2006kernel, kim2021comparing}. However, these bounds tend to be overly conservative and result in tests with low  \citep{wolfer2025variance}. \cite{balsubramani2016sequential} proposed a sequential two-sample test based on the Maximum Mean Discrepancy (MMD) with a linear kernel, using martingale concentration inequalities. More recently, \cite{shekhar2023nonparametric} extended this approach by introducing a more general sequential MMD-based test that leverages Ville's inequality, resulting in a less conservative procedure. \citet{gretton2012kernel} proposed moment-based parametric approximations such as Pearson curves and Gamma distributions, but these remain heuristic in nature and lack rigorous guarantees. \citet{gretton2009fast} introduced a spectral method for approximating the asymptotic distribution, which required strong assumptions and needed additional resampling (with external randomness) for implementation.

To address computational inefficiency, later works proposed modifying the test statistic itself to yield a tractable null distribution. One such method partitions the data into disjoint blocks, computes the kernel MMD statistic within each block, and averages the results. When block sizes are fixed, this leads to the linear-time MMD test \citep{gretton2012kernel}. \citet{zaremba2013b,ramdas2015adaptivity,reddi2015high} extended this idea to block sizes that grow with the sample size. The computational complexity of the block-MMD statistic then ranges from linear (for constant block size) to quadratic (when block size is $\Omega(n)$). If the block size  $B=o(n)$, the test statistic admits a Gaussian limit under the null. These approaches fall under the umbrella of incomplete U-statistcs based approaches. We refer to \cite{schrab2022efficient} for a comprehensive overview of incomplete U-statistic-based kernel tests. We empirically compare our method with both linear-time and block-MMD tests in Section~\ref{sec:power_sims}.

Recent works have also explored sample-splitting strategies. Notably, \citet{kubler2022witness} used one half of the data to estimate the witness function and the other half to compute the test statistic. However, their analysis assumes the witness function is fixed and relies on permutation testing for calibration, which limits practical efficiency. In contrast, our test leverages asymptotic Gaussianity to avoid resampling, making it more straightforward to implement. 

\citet{shekhar2022permutation} proposed a sample-splitting approach, where in each split they compute the difference between the mean embeddings and apply RKHS inner product between the two estimated embedding differences. Their dimension-agnostic Gaussian null distribution is something that we also achieve, but we also prove that our statistic has an alternative distribution that is Gaussian. We show empirically that our approach performs comparable to theirs, but ours avoids sample splitting and is based on different principles.

Beyond methods aimed at efficient threshold selection, several works have addressed the computational cost of the test statistic itself. The standard quadratic-time MMD statistic scales as $O(n^2)$. To reduce this, methods based on Random Fourier Features \citep{zhao2015fastmmd, zhao2022comparing, choi2024computational, mukherjee2025minimax}, coreset-based approximations \citep{domingo2023compress}, and Nyström approximations \citep{chatalic2025efficient} have been proposed. However, most of these still rely on permutation or resampling techniques to calibrate the test statistic, which remains the main focus of this work.

\section{The martingale Maximum Mean Discrepancy $(m\MMD)$ Test}

\subsection{Preliminaries}
In this section, we review the fundamentals of kernel two-sample testing as introduced by~\cite{gretton2012kernel}, and present the key concepts required for our proposed test. We begin by establishing the necessary formalism.\\

\noindent
Let $\cX$ be a Polish space, that is a complete seperable metric space, equipped with the corresponding Borel sigma algebra $\cB_\cX$. Let $\cP(\cX)$ be the collection of all probability measures defined on $(\cX, \cB_\cX)$. Let $\sfK$ be a positive definite kernel with features maps $\phi_x, x\in \cX$ such that $\sfK(x,\cdot) = \phi_x$ and $\sfK(x_1,x_2) = \langle \phi_{x_1},\phi_{x_2}\rangle_{\cK}$ where $\cK$ is the associated RKHS. The notion of feature maps can be extended to embed any distribution $P\in \cP(\cX)$ into $\cK$. In particular, for any $P\in \cP(\cX)$ define the kernel mean embedding $\nu_P$ as $\langle f, \nu_P\rangle = \E_{X\sim P}[f(X)]$ for all $f\in \cK$. Moreover, by the cannonical form of the feature maps, for all $t\in \cX, \nu_P(t) = \E_{X\sim P}[\sfK(X, t)]$. Throughout the rest of the article we make the following assumption on the kernel $\sfK$.
\begin{assumption}\label{assumption:K}
    Consider the positive definite kernel $\sfK:\cX\times \cX\ra \R$ and let $\cM_\sfK^{\theta} = \{P\in \cP(\cX): \E_{X\sim P}[\sfK(X,X)^\theta]<\infty\}$. Then $\sfK$ is characteristic with respect to $\cM_{\sfK}^{1/2}$, that is $\nu:\cM_{\sfK}^{1/2}\ra\cK$ is a one-to-one (injective) mapping.
\end{assumption}
Assumption~\ref{assumption:K} ensures that $\mu_P\in \cK$ for all $P\in \cM_{\sfK}^{1/2}$ (see \citet[Lemma 3]{gretton2012kernel} and \citet[Lemma 2.1]{park2020measure}). With the above definitions we can now define the Maximum Mean Discrepancy (MMD) between two distribution $P,Q\in \cM_{\sfK}^{1/2}$ as,
\begin{align}\label{eq:def_MMD}
    \MMD[P,Q, \sfK] = \sup_{f\in\cF_\cK}\E_{X\sim P}[f(X)] - \E_{Y\sim Q}[f(Y)],
\end{align}
where $\cF_\cK$ is the unit-ball in the RKHS $\cK$. Assumption~\ref{assumption:K} also ensures that $\MMD$ forms a metric on $\cM_{\sfK}^{1/2}$. Moreover, using the notion of mean embeddings the $\MMD$ can be equivalently expressed as $\MMD[P,Q, \sfK] = \|\nu_P - \nu_Q\|_\cK$ where $\|\cdot\|_\cK$ is the norm induced by the inner product $\langle\cdot, \cdot\rangle_{\cK}$. 

\subsection{Deriving the $m\MMD$ test}\label{sec:derive_mMMD_test}
The standard quadratic-time $\MMD$ test statistic is obtained by substituting the empirical mean embeddings into the alternative expression for squared $\MMD[P, Q]$ derived earlier. In contrast, our approach returns to the original definition of $\MMD[P, Q]$ given in~\eqref{eq:def_MMD}.

\begin{figure}[htbp]
    \centering
    \begin{minipage}[b]{0.24\linewidth}
        \centering
        \includegraphics[width=\linewidth]{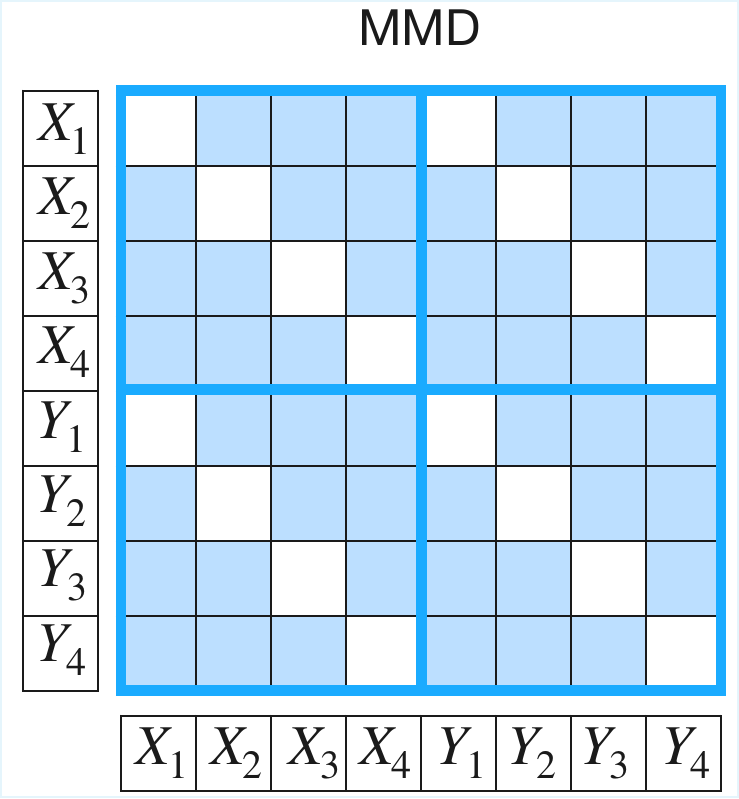}
    \end{minipage}
    \hfill
    \begin{minipage}[b]{0.24\linewidth}
        \centering
        \includegraphics[width=\linewidth]{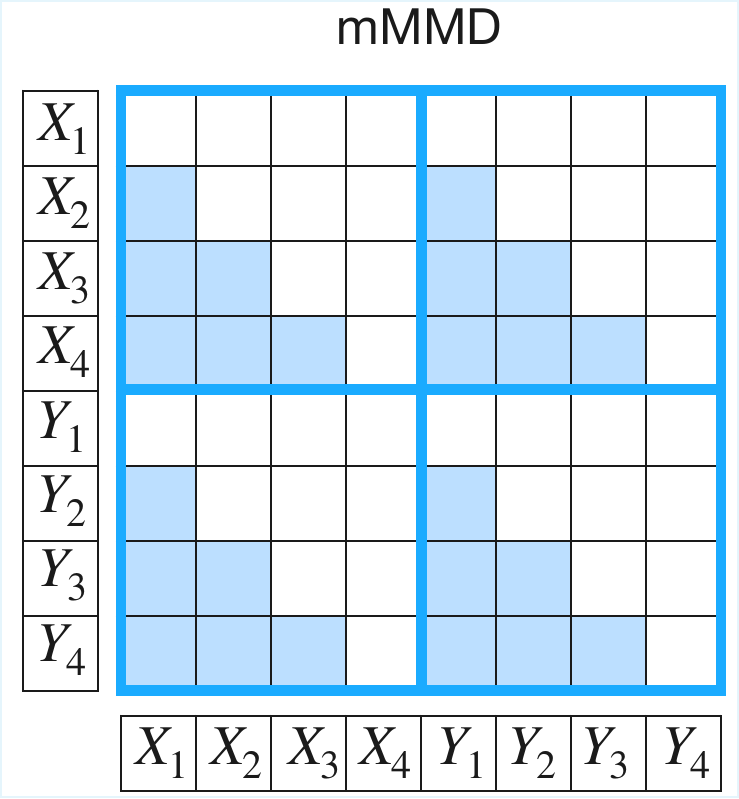}
    \end{minipage}
    \hfill
    \begin{minipage}[b]{0.24\linewidth}
        \centering
        \includegraphics[width=\linewidth]{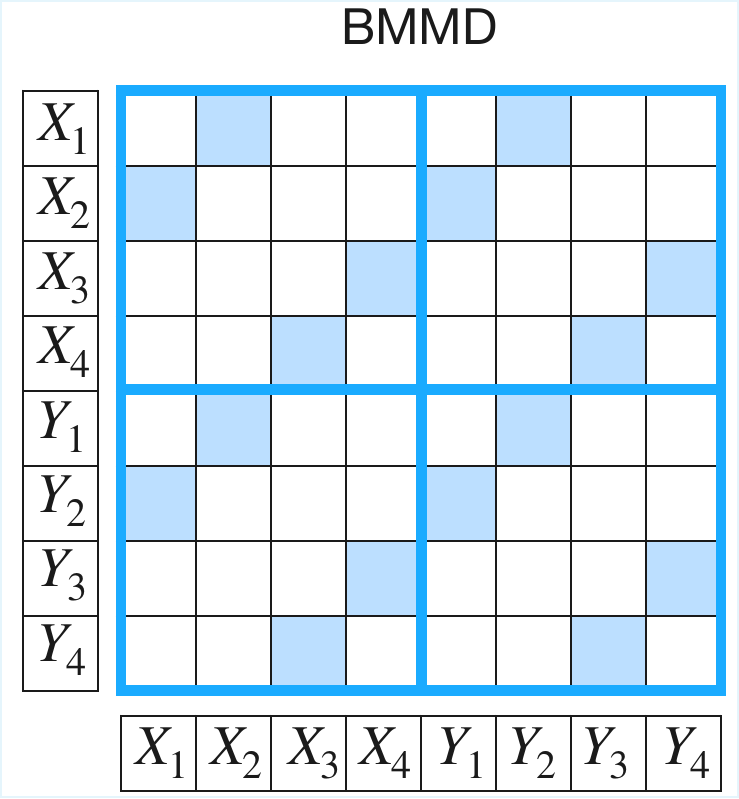}
    \end{minipage}
    \hfill
    \begin{minipage}[b]{0.24\linewidth}
        \centering
        \includegraphics[width=\linewidth]{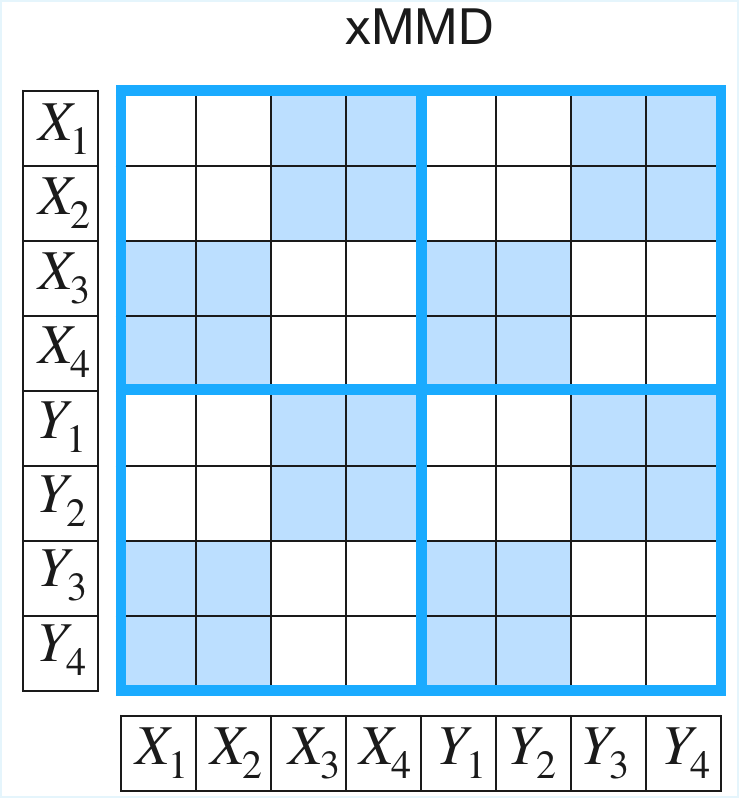}
    \end{minipage}
    
    \caption{A visual illustration of the main difference in computing the quadratic-time MMD (\texttt{MMD}), our proposed \texttt{mMMD}, the block-MMD (\texttt{BMMD}) from \citet{zaremba2013b}, and the cross-MMD (\texttt{xMMD}) from \citet{shekhar2022permutation}. The \textcolor{cyan}{highlighted} regions indicate the sample pairs used in the computation. The figure shows that the quadratic-time MMD considers all off-diagonal pairwise kernel evaluations between the combined samples. \texttt{BMMD} partitions the data into blocks and averages the pairwise kernel evaluations within each block. \texttt{xMMD} splits the samples into two halves and evaluates kernels only across the splits. In contrast, our proposed \texttt{mMMD} computes pairwise kernel values by taking the lower triangular part of within-sample and between-sample blocks. However, we emphasize that, due to the symmetry of the kernel, the terms used by \texttt{MMD} and \texttt{mMMD} are equivalent. The principal distinction between the two lies in the normalization procedure: \texttt{mMMD} computes a row-wise mean followed by an average across rows (see \eqref{eq:Tn_alt}), whereas \texttt{MMD} applies a global mean over all elements.}
    \label{fig:mmd_variants}
\end{figure}

\noindent
Note that if the witness function $f_0$, that is, $f_0 = \arg\sup_{f \in \cF} \left\{ \E_{X \sim P}[f(X)] - \E_{Y \sim Q}[f(Y)] \right\}$, were known, then one could estimate $\MMD[P, Q]$ using the statistic $\frac{1}{n} \sum_{i=1}^n \left\langle f_0, \sfK(X_i, \cdot) - \sfK(Y_i, \cdot) \right\rangle_{\cK},$ based on samples $\sX_n = \{X_1, \ldots, X_n\} \sim P$ and $\sY_n = \{Y_1, \ldots, Y_n\} \sim Q$. However, the witness function $f_0$ is unknown. Nevertheless, as shown by \citet[Section~2.3]{gretton2012kernel} it can be expressed as $f_0 = (\nu_P - \nu_Q)/\|\nu_P - \nu_Q\|_{\cK}$, where $\|\cdot\|_{\cK}$ denotes the RKHS norm. Then for each $2\leq i\leq n$ we can plug in the empirical mean embeddings to approximate the unnormalized version of $f_0$ by $\hat f_i = \frac{1}{i} \sum_{j=1}^{i-1} \left[ \sfK(X_j, \cdot) - \sfK(Y_j, \cdot) \right]$, which leads to our proposed statistic:
\begin{align*}
    T_n := \frac{1}{n} \sum_{i=2}^n \left\langle \hat f_i, \sfK(X_i, \cdot) - \sfK(Y_i, \cdot) \right\rangle_{\cK}.
\end{align*}
\noindent
To further motivate this statistic, observe that by using only past data (notice that we only use indices less than $i$ for the estimate $\hat f_i$) to estimate the unnormalized witness function $f_0$, we induce an underlying martingale structure. In particular, under $\bH_0$, the collection 
\begin{align}\label{eq:martingale_seq_Tn}
    \left\{ \left\langle \hat f_i, \sfK(X_i, \cdot) - \sfK(Y_i, \cdot) \right\rangle_{\cK} : i \geq 2 \right\}
\end{align}
forms a martingale difference sequence with respect to the natural filtration $\cF_i = \sigma(\bZ_1, \ldots, \bZ_i)$, where $\bZ_i = (X_i, Y_i), 2\leq i\leq n$. With this martingale formulation one can expect to apply the martingale CLT \citep{hall2014martingale} on $T_n$ (with proper normalisation) to have asymptotic Gaussianity. However, directly applying a martingale central limit theorem (CLT) to $T_n$ may yield an asymptotic Gaussian distribution with unknown variance. Therefore, to implement a test for~\eqref{eq:H01} based on $T_n$, we must appropriately normalize the statistic. To that end, define
\begin{align}\label{eq:sigma_n_exp}
    \sigma_n^2 := \frac{1}{n^2}\sum_{i=2}^{n} \left( \frac{1}{i} \sum_{j=1}^{i-1} \sfK(X_i, X_j) - \sfK(X_i, Y_j) - \sfK(X_j, Y_i) + \sfK(Y_i, Y_j) \right)^2.
\end{align}
With this notation, we define the normalized test statistic as
\begin{align}\label{eq:def_eta_n}
    \eta_n = T_n/\sigma_n,
\end{align}
and use it to construct the $m\MMD$ test as described in~\eqref{eq:mMMD_test}.

\begin{remark}
    Note that the computational complexity of computing $\eta_n$ is $O(n^2)$. To see this, observe that using the kernel trick, the statistic $T_n$ can be simplified as
    \begin{align}\label{eq:Tn_alt}
        T_n = \frac{1}{n} \sum_{i=2}^{n} \frac{1}{i} \sum_{j=1}^{i-1} \sfK(X_i, X_j) - \sfK(X_i, Y_j) - \sfK(X_j, Y_i) + \sfK(Y_i, Y_j).
    \end{align}
    It is evident from this expression that computing $T_n$ requires $O(n^2)$ operations. Moreover, from the definition of $\sigma_n$ in~\eqref{eq:sigma_n_exp}, it follows that computing $\sigma_n$ also requires $O(n^2)$ time. Hence, the total cost of computing the test statistic $\eta_n$ is $O(n^2)$. This matches the computational complexity of the original quadratic-time $\MMD$ statistic proposed by~\citet{gretton2012kernel}, as well as the more recent cross-MMD statistic introduced by~\citet{shekhar2022permutation}. In contrast, the alternative approaches proposed by~\citet{zaremba2013b} and~\citet{jitkrittum2017linear} can achieve lower computational complexity, at the expense of highly suboptimal power.

\end{remark}

\section{Asymptotic Null Distribution of $m\MMD$ Test Statistic}\label{sec:null_dist_section}
In this section, we derive the limiting distribution of the statistic $\eta_n$ (as defined in \eqref{eq:def_eta_n}) under the null hypothesis $\bH_0$. Specifically, we establish the asymptotic normality of $\eta_n$, a phenomenon that is empirically suggested by Figure~\ref{fig:overview} (a).\\

\noindent
For a kernel $\sfK$, and define its centered version as $\bsfK(x, y) = \left\langle \sfK(x, \cdot) - \nu_P,\ \sfK(y, \cdot) - \nu_P \right\rangle_{\cK}$, where $\nu_P$ denote the kernel mean embedding of $P$.

\begin{theorem}\label{thm:null_clt_fixed}
    Take $\cX = \R^d$ for some $d\geq 1$. Suppose the kernel $\sfK$ satisfies Assumption~\ref{assumption:K}, and let $P\in \cM_{\sfK}^{1/2}$. Moreover assume $\E\left[\bsfK\left(X_1,X_2\right)^4\right]<\infty$ for $X_1,X_2$ generated independently from $P$. Then under $\bH_0$,
    \begin{align*}
        \sup_{x\in \R}\left|\P\left(\eta_n\leq x\right) - \Phi(x)\right|\ra 0.
    \end{align*}
\end{theorem}

The above theorem confirms that, under $\bH_0$, the asymptotic distribution of the $m\MMD$ test statistic $\eta_n$ is standard Gaussian. Consequently, the test $\Psi_n$ defined in~\eqref{eq:mMMD_test} asymptotically controls the type-I error at level $\alpha$.

\begin{remark}[Martingale Energy Two-Sample Test]
Take $\cX = \R^d$ and consider the distance induced kernel
$\sfK(x,y) = \frac{1}{2}\big(\|x\|_2 + \|y\|_2 - \|x-y\|_2\big)$, so that the squared maximum mean discrepancy satisfies $\MMD^2[P,Q;\sfK] = \frac{1}{2}\rho(P,Q)$, where $\rho(P,Q)$ denotes the classical energy distance between $P$ and $Q$ \citep{szekely2003statistics,szekely2013energy, sejdinovic2013equivalence}. This kernel is characteristic with respect to $\cM_{\sfK}^1$, providing an RKHS representation of the energy distance. With this choice of kernel, we can define a \emph{martingale energy two-sample test} using the normalised martingale energy statistic as,
\begin{align*}
\eta_{n,\rho} = \frac{
\frac{1}{n} \sum_{i=2}^{n} \frac{1}{i} \sum_{j=1}^{i-1} 
\|X_i - Y_j\|_2 + \|X_j - Y_i\|_2 - \|X_i - X_j\|_2 - \|Y_i - Y_j\|_2
}{
\sqrt{
\frac{1}{n^2} \sum_{i=2}^{n} 
\Big( \frac{1}{i} \sum_{j=1}^{i-1} 
\|X_i - Y_j\|_2 + \|X_j - Y_i\|_2 - \|X_i - X_j\|_2 - \|Y_i - Y_j\|_2
\Big)^2 }
}.
\end{align*}
Note that assumptions of Theorem \ref{thm:null_clt_fixed} are satisfied whenever $\mathbb{E}\left[\|X\|_2^4\right] < \infty$, and under the null hypothesis it is asymptotically standard normal. This allows us to define the martingale energy test $\Psi_{n,\rho} = \mathbf{1}\{\eta_{n,\rho} > z_{1-\alpha}\}$, where (as before) $z_{1-\alpha}$ denotes the $(1-\alpha)$-quantile of the standard normal distribution.
\end{remark}

\noindent
The above result can be significantly extended to accommodate more general dependencies of the kernel, the distribution, and even the sample space on the sample size $n$. Specifically, we now allow the kernel $\mathsf{K}_n$, the null distribution $P_n (= Q_n)$, and the dimension $d_n$ to vary with $n$\footnote{To denote dependence of a parameter $\theta$ on $n$ we use the notation $\theta_n$.}. In this more general setting, the next result shows that, under suitable conditions, the test statistic $\eta_n$ continues to converge in distribution to a standard Gaussian under the null.

\begin{theorem}\label{thm:null_clt_general}
    Take $n\geq 1$ and $\cX_n = \R^{d_n}$. Suppose the kernel $\sfK_n$ satisfies Assumption~\ref{assumption:K} on $\cX_n$ and $P_n\in \cM_{\sfK_n}^{1/2}$. With $X_1,X_2,X_3$ generated independently from $P_n$ assume that $\E\left[\bsfK_n(X_1, X_2)^4\right]<\infty$. Furthermore assume the following holds,
    \begin{align}\label{eq:null_clt_cond2}
        \lim_{n\rightarrow \infty}\frac{\E\left[\bsfK_n(X_1, X_2)^4\right]}{(\log n)^2\E\left[\bsfK_n(X_1, X_2)^2\right]^2} = 0\text{ and }\lim_{n\rightarrow \infty}\frac{\E\left[\bsfK_n(X_1, X_2)^2\bsfK_n(X_1, X_3)^2\right]}{(\log n)^2\E\left[\bsfK_n(X_1, X_2)^2\right]^2} = 0.
    \end{align}
    Then under $\bH_0$,
    \begin{align*}
        \sup_{x\in \R}\left|\P\left(\eta_n\leq x\right) - \Phi(x)\right|\ra 0.
    \end{align*}
\end{theorem}
\noindent
Note that the condition in \eqref{eq:null_clt_cond2} is well-defined by Lemma~\ref{lemma:E_barK_sq_positive}. Moreover, Theorem~\ref{thm:null_clt_fixed} follows as an immediate corollary of Theorem~\ref{thm:null_clt_general}, which in turn can be viewed as a special case of Corollary~\ref{cor:NullCLT} by setting $\gamma = 1$. The proof of Theorem~\ref{thm:null_clt_general} (and more generally Theorem~\ref{thm:berry_essen_H0}) proceeds by recalling the martingale difference sequence in the construction of $\eta_n$ from \eqref{eq:martingale_seq_Tn}, and applying Berry--Esseen bounds for self-normalized martingales from \cite{fan2018berry}.

\begin{remark}
    In the case where $d_n = d$ is fixed, the condition \eqref{eq:null_clt_cond2} is satisfied for uniformly bounded kernels, including commonly used examples such as Gaussian and Laplace kernels, when appropriate $n$-dependent bandwidths are used (see Section~\ref{sec:proof_minimax} and \cite{li2024optimality}). Similar conditions have been considered in the context of kernel two-sample tests by \cite{shekhar2022permutation} and \cite{li2024optimality}. Furthermore, a condition akin to \eqref{eq:null_clt_cond2} has appeared in the broader literature on studentized central limit theorems, notably in \cite{bentkus1996berry}, in the setting of triangular arrays of i.i.d.\ random variables. While their assumption is phrased in terms of third moments, our condition is slightly stronger, due to the fact that the summands in our case form a martingale difference sequence and thus exhibit dependencies not present in the i.i.d.\ setting.
\end{remark}

\begin{remark}\label{remark:non_median_band}
    Note that Theorem~\ref{thm:null_clt_general} permits the kernel to depend on the sample size $n$, but not on the observed samples $\sX_n$ and $\sY_n$. This is particularly relevant as it excludes the common practice of selecting kernel bandwidths, such as for Gaussian or Laplace kernels, based on the median of pairwise distances from the pooled samples (the median heuristic). Nevertheless, we find in our experiments that using the median heuristic retains the same null distribution, so it is possible that the above theoretical guarantees can be extended to certain data-dependent kernels as well.
\end{remark}

\section{Consistency of the $m\MMD$ Test}
In this section, we show that the $m\MMD$ test defined in \eqref{eq:mMMD_test} is consistent, that is, it achieves asymptotic power one, against fixed alternatives. Furthermore, we establish conditions under which the test achieves \emph{uniform consistency} against a class of alternatives in the more challenging setting where the kernel, distributions, and sample space may all vary with the sample size $n$.\\

\noindent
We first show that the $m\MMD$ test is consistent against fixed alternatives.
\begin{theorem}\label{thm:fixed_consistency}
    Take $\cX = \R^d$ for some $d\geq 1$ and for a kernel $\sfK$ adopt Assumption~\ref{assumption:K}. Let $P,Q\in \cM_{\sfK}^{1}$ such that $P\neq Q$. Then $\E[1-\Psi_n]\ra 1$, that is the test $\Psi_n$ from \eqref{eq:mMMD_test} is consistent.
\end{theorem}
\noindent
The above theorem requires only mild assumptions, which are trivially satisfied by commonly used kernels such as the Gaussian and Laplace kernels. Analogous to the setting discussed in Theorem~\ref{thm:berry_essen_H0}, the following result establishes conditions under which consistency continues to hold even when the kernel, distributions, and sample space are allowed to vary with the sample size $n$.

\begin{theorem}\label{thm:gen_consistency}
    Take $n\geq 1$ and $\cX_n = \R^{d_n}$. Consider a positive definite kernel $\sfK_n$ satisfying Assumption~\ref{assumption:K} on $\cX_n$ and take $\cP_n\subseteq\cM_{\sfK_n}^{1/2}\times \cM_{\sfK_n}^{1/2}$. For $(P_n,Q_n)\in \cP_n$ assume that $\delta_n := \MMD[P_n, Q_n,\sfK_n]>0$. If,
    \small
    \begin{align*}
        \lim_{n\ra\infty} \sup_{(P_n, Q_n)\in \cP_n} \frac{\E[\sigma_n^2]}{\delta_n^4} + \frac{\Var\left(T_n\right)}{\delta_n^4} = 0,
    \end{align*}
    \normalsize
    then $\lim_{n\ra\infty}\sup_{(P_n, Q_n)\in \cP_n}\E\left[1-\Psi_n\right] = 0$, that is $\Psi_n$ is consistent.
\end{theorem}
\noindent
Theorem~\ref{thm:gen_consistency} is an immediate consequence of Theorem~\ref{thm:general_consistency}. Theorem~\ref{thm:fixed_consistency} indeed follows by checking the conditions of Theorem~\ref{thm:gen_consistency}, which we detail in Section~\ref{sec:proofof_fixed_consistency}. We note that the conditions required for the consistency of the test $\Psi_n$ are identical to those needed for the consistency of the cross-MMD test proposed by \cite{shekhar2022permutation}.

\section{Distribution under the alternative}
In this section, we show that the $m$-MMD statistic $T_n$ converges to a Gaussian distribution even under the alternative hypothesis, that is, when $P \neq Q$. To state the result, we consider the setting, allowing the distributions $P_n$, $Q_n$, and the kernel $\sfK_n$ to vary with the sample size $n$. However note that in contrast to the previous general settings here $\cX = \R^{d}$ remains fixed with respect to $n$.\\

\noindent
We begin by introducing the following notation:
\begin{align}\label{eq:def_Hn_main}
    \sfH_n(\bz_1, \bz_2) := \sfK_n(x_1, x_2) - \sfK_n(x_1, y_2) - \sfK_n(x_2, y_1) + \sfK_n(y_1, y_2),
\end{align}
where $\bz_i = (x_i, y_i)$ for $i = 1, 2$. Furthermore, for $\bZ = (X, Y) \sim P_n \times Q_n$, define
\begin{align}\label{eq:def_h_n}
    h_n(\bz) := \E\left[\sfH_n(\bz, \bZ)\right] - \MMD[P_n, Q_n, \sfK_n]^2.
\end{align}

With this setup in place, we now state the result characterizing the asymptotic distribution of $T_n$ under the alternative $P_n \neq Q_n$.

\begin{theorem}\label{thm:general_alt_dist}
    Adopt Assumption~\ref{assumption:K} with $\sfK_n$. Let $P_n, Q_n\in \cM_{\sfK_n}^{3/2}$ such that $P_n\neq Q_n$ and assume that for $\bZ,\bZ_1,\bZ_2\sim P_n\times Q_n$,
    \begin{align*}
        \frac{\E\left[\left|h_{n}(\bZ)\right|^3\right]}{\Var\left(h_n(\bZ)\right)^{3/2}}\ll \frac{n^{3/2}}{(\log n)^3},\ \frac{\E\left[g_n(\bZ_1,\bZ_2)^2\right]}{\Var\left(h_{n}(\bZ)\right)}\ll \frac{n}{\log n},\ \frac{\MMD[P_n, Q_n, \sfK_n]^2}{\sqrt{\Var\left(h_{n}(\bZ)^2\right)}}\ll\frac{\sqrt{n}}{\log n},
    \end{align*}
    where $g_n(\bz_1,\bz_2) = \sfH_n(\bz_1,\bz_2) - h_n(\bz_1) - h_n(\bz_2) + \MMD[P_n,Q_n,\sfK_n]^2$. Then,
    \begin{align*}
        \sqrt{\frac{n}{\Var\left(h_{n}(\bZ)\right)}}\left(T_{n}-\MMD[P_n,Q_n,\sfK_n]^2\right)\dto \mathrm N(0,5).
    \end{align*}
\end{theorem}
Using the above theorem, when $P,Q$ and $\sfK$ does not depend on $n$ we can easily conclude that,
\begin{align}\label{eq:alt_dist}
    \sqrt{n}\left(T_n - \MMD[P,Q,\sfK]^2\right)\dto \mathrm N\left(0, 5\Var\left(h(\bZ)\right)\right)\text{ under }\bH_1,
\end{align}
where $h(\bZ)$ is defined in \eqref{eq:def_h_n} (without the dependence on $n$). We note that the asymptotic distribution in \eqref{eq:alt_dist} closely resembles that of the original quadratic-time MMD estimator from \citet{gretton2012kernel}, having a constant factor 5 in the asymptotic variance in place of 4. While the results in Section~\ref{sec:null_dist_section} are similar to those for the cross-MMD test by \citet{shekhar2022permutation}, the latter do not establish distributional convergence under the alternative.

\section{Experiments}\label{sec:experiments}
We empirically evaluate the proposed $m\MMD$-test against the linear-time MMD test \citep{gretton2012kernel}, Block MMD with $B = \lfloor\sqrt{n}\rfloor$ \citep{zaremba2013b}, and cross-MMD \citep{shekhar2022permutation}, collectively referred to as computationally efficient MMD tests. Section~\ref{sec:empirical_null} validates the null distribution from Section~\ref{sec:null_dist_section}, while Sections~\ref{sec:power_sims} and~\ref{sec:real_data} compare performance on simulated and real data. In all experiments, we use the median heuristic for bandwidth selection. Although this choice is not supported by our theory (see Remark~\ref{remark:non_median_band}), nor by the theoretical guarantees in most prior works, it remains the standard in empirical studies due to its strong practical performance. Code for all experiments in this work is available at  \href{https://github.com/anirbanc96/mMMD}{https://github.com/anirbanc96/mMMD}. 

\subsection{Limiting Null Distribution of $\eta_n$}\label{sec:empirical_null}
In Theorem~\ref{thm:null_clt_fixed}, we established that the asymptotic null distribution of the $m$-MMD test statistic $\eta_n$ is standard Gaussian under mild conditions. In this section, we empirically validate this result across a range of settings, examining the impact of dimensionality, data distribution, and kernel choice. Specifically, we consider dimensions $d \in \{10, 100, 250, 500\}$; data distributions $\mathrm N_d(\mathbf{0}_d, \mathbf{I}_d)$ and $t_d(10)$; and both Gaussian and Laplace kernels. All experiments in this section use a sample size of $n = 200$ and $2000$ resamples for the empirical distributions.

\begin{figure}[h!]
  \centering

  \begin{minipage}[b]{0.24\textwidth}
    \centering
    \includegraphics[width=\linewidth]{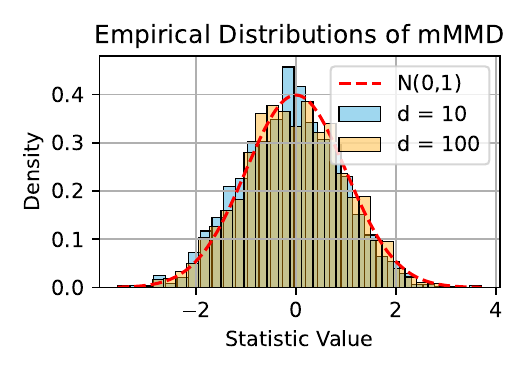}
  \end{minipage}
  \hfill
  \begin{minipage}[b]{0.24\textwidth}
    \centering
    \includegraphics[width=\linewidth]{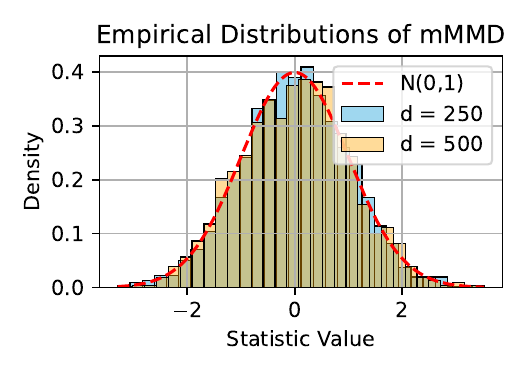}
  \end{minipage}
  \hfill
  \begin{minipage}[b]{0.24\textwidth}
    \centering
    \includegraphics[width=\linewidth]{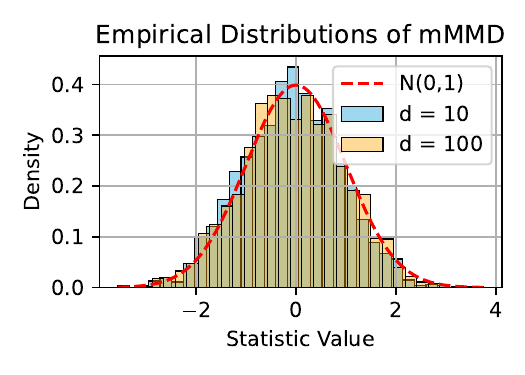}
  \end{minipage}
  \hfill
  \begin{minipage}[b]{0.24\textwidth}
    \centering
    \includegraphics[width=\linewidth]{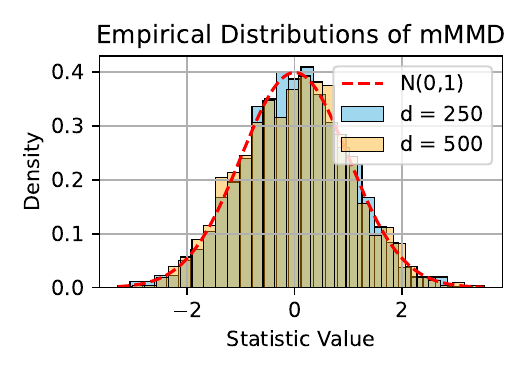}
  \end{minipage}

  \caption{Empirical null distribution of the $m$-MMD test statistic $\eta_n$ from equation~\eqref{eq:def_eta_n}. The left two figures use a Gaussian kernel with the median heuristic, while the right two figures use a Laplace kernel, also with the median heuristic. The underlying data distribution is $\mathrm{N}(\mathbf{0}_d, \mathbf{I}_d)$.}
  \label{fig:null_dist_gauss}
\end{figure}

\begin{figure}[h!]
  \centering

  \begin{minipage}[b]{0.24\textwidth}
    \centering
    \includegraphics[width=\linewidth]{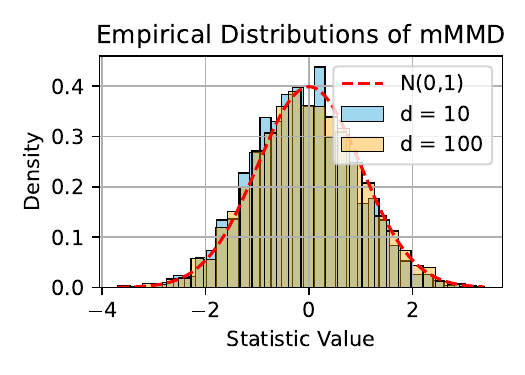}
  \end{minipage}
  \hfill
  \begin{minipage}[b]{0.24\textwidth}
    \centering
    \includegraphics[width=\linewidth]{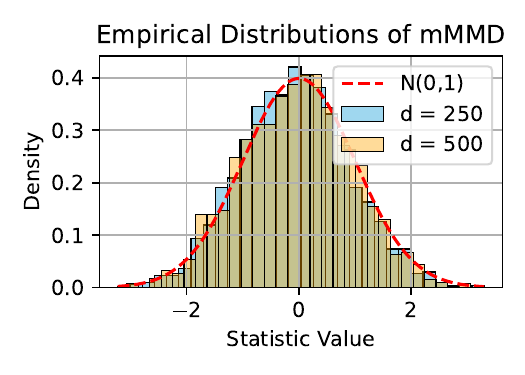}
  \end{minipage}
  \hfill
  \begin{minipage}[b]{0.24\textwidth}
    \centering
    \includegraphics[width=\linewidth]{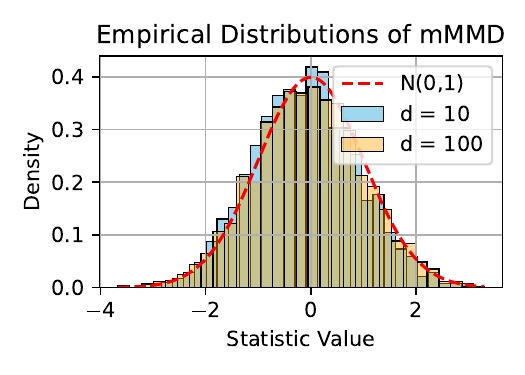}
  \end{minipage}
  \hfill
  \begin{minipage}[b]{0.24\textwidth}
    \centering
    \includegraphics[width=\linewidth]{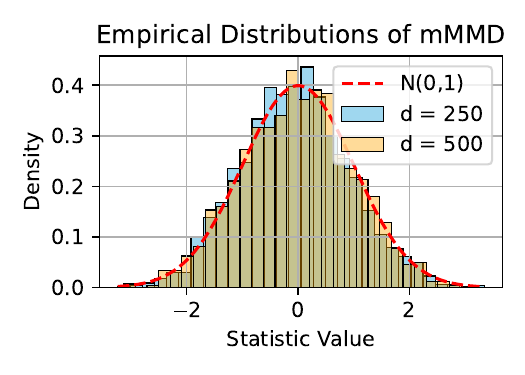}
  \end{minipage}

  \caption{Empirical null distribution of the $m$-MMD test statistic $\eta_n$ from equation~\eqref{eq:def_eta_n}. The left two figures use a Gaussian kernel with the median heuristic, while the right two use a Laplace kernel, also with the median heuristic. The underlying data distribution is $t_d(10)$.}
  \label{fig:null_dist_t}
\end{figure}

In Figure~\ref{fig:null_dist_gauss}, we plot the empirical null distribution of $\eta_n$ from equation~\eqref{eq:def_eta_n}, where the underlying data distribution is $\mathrm{N}_d(\bm{0}_d, \bm{I}_d)$. In Figure~\ref{fig:null_dist_t}, we consider the same setup but with multivariate $t$-distribution having degrees of freedom 10 as the underlying data.

Across all these settings, the empirical distribution of $\eta_n$ closely matches the standard Gaussian, demonstrating robustness to the data distribution, kernel choice, and dimensionality. This stands in sharp contrast to the null distribution of the quadratic-time MMD statistic from \cite{gretton2012kernel}, which varies significantly, particularly with increasing dimension  \citep{shekhar2022permutation}.

\subsection{Power Experiments}\label{sec:power_sims}
In this section, we compare the performance of our proposed test with the quadratic-time MMD test with 200 permutations, the sequential MMD test from \cite{shekhar2023nonparametric} (referred as \texttt{BetMMD}) and the computationally efficient MMD tests.

For the experimental setup, let $d \geq 1$, and define $\bm\mu_{d,j,\varepsilon} \in \mathbb{R}^d$ such that the first $j$ coordinates are equal to $\varepsilon$ and the remaining coordinates are zero. We consider a two-sample testing problem with distributions $P = \mathrm{N}_d(\mathbf{0}_d, \mathbf{I}_d)$ and $Q = \mathrm{N}_d(\bm\mu_{d,j,\varepsilon}, \mathbf{I}_d)$. The specific configurations we evaluate are $(d, j, \varepsilon) = (10, 5, 0.3)$, $(50, 5, 0.3)$, and $(100, 5, 0.5)$. This experimental setting has also been considered in \cite{shekhar2022permutation}. All experiments in this section use the Gaussian kernel.

\begin{figure}[h!]
  \centering

  \begin{minipage}[b]{0.32\textwidth}
    \centering
    \includegraphics[width=\linewidth]{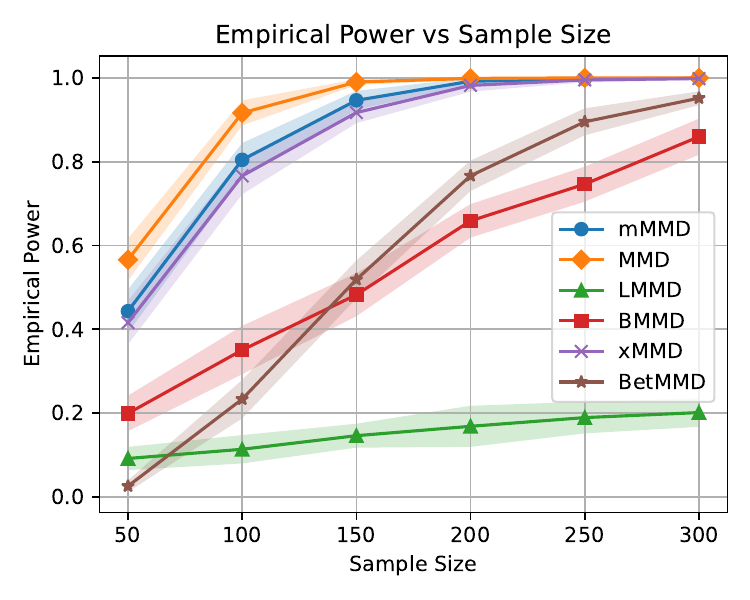}
  \end{minipage}
  \hfill
  \begin{minipage}[b]{0.32\textwidth}
    \centering
    \includegraphics[width=\linewidth]{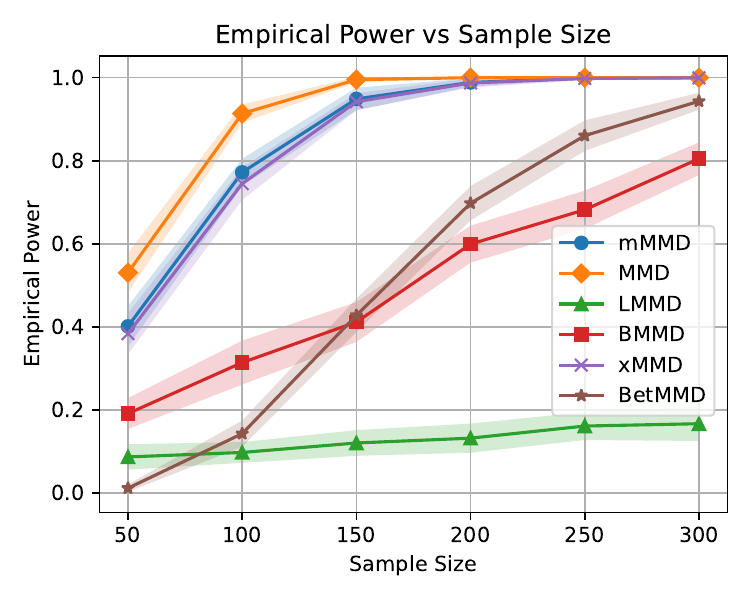}
  \end{minipage}
  \hfill
  \begin{minipage}[b]{0.32\textwidth}
    \centering
    \includegraphics[width=\linewidth]{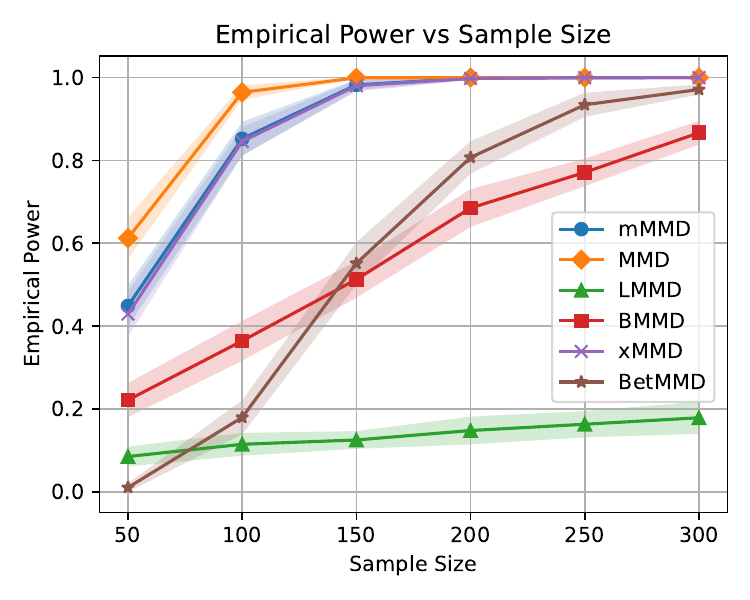}
  \end{minipage}

  \caption{Empirical power of the $m$-MMD test compared to the original quadratic-time \texttt{MMD}, the linear-time \texttt{LMMD}, the block-based \texttt{BMMD}, and the cross-MMD (\texttt{xMMD}) tests. From left to right, the figures correspond to $(d, j, \varepsilon) = (10, 5, 0.3)$, $(50, 5, 0.4)$, and $(100, 5, 0.5)$.}
  \label{fig:power_comparison}
\end{figure}

We present the performance comparison in Figure~\ref{fig:power_comparison}. Across all settings, the $m$-MMD test consistently outperforms both the linear-time MMD and Block MMD tests, and achieves performance comparable to the cross-MMD test. While its performance is slightly below that of the quadratic-time MMD test, this is expected due to the trade-off between statistical  and computational efficiency.

\subsection{Real Data Experiment}\label{sec:real_data}

In this section we compare the performance of our proposed test with computationally efficient MMD tests on the MNIST \citep{lecun1998mnist} data set. We evaluate the performance of tests for the test from \eqref{eq:H01} with the following $5$ pairs of sets of digits: \texttt{Group 1}: $P = \{2, 4, 8, 9\}$ and $Q = \{3, 4, 7, 9\}$, \texttt{Group 2}: $P = \{1, 2, 4, 8, 9\}$ and $Q = \{1, 3, 4, 7, 9\}$, \texttt{Group 3}: $P = \{0, 1, 2, 4, 8, 9\}$ and $Q = \{0, 1, 3, 4, 7, 9\}$, \texttt{Group 4}: $P = \{0, 1, 2, 4, 5, 8, 9\}$ and\\ $Q = \{0, 1, 3, 4, 5, 7, 9\}$, and \texttt{Group 5}: $P = \{0, 1, 2, 4, 5, 6, 8, 9\}$ and $Q = \{0, 1, 3, 4, 5, 6, 7, 9\}$. Similar experimental settings have been considered in \cite{schrab2023mmd} and \cite{chatterjee2025boosting}. 

\begin{figure}[H]
  \centering

  \begin{minipage}[t]{0.5\textwidth}
  The entire procedure is repeated 100 times, and the resulting empirical power is reported in Figure~\ref{fig:mnist_asymptotic}.

\noindent
For each of the five methods, we draw $100$ samples with replacement from the sets $P$ and $Q$ and use the asymptotic Gaussian distribution to calibrate the test under $\bm{H}_0$ at a significance level of $0.05$. Notice that by construction, with increasing group number it becomes harder to detect difference between the distributions. Indeed in Figure~\ref{fig:mnist_asymptotic} we note that the power of all methods decreases with increasing group index. However for the proposed $m\MMD$-test and the cross-MMD tests the power decrement is much slower and overall they perform significantly better than other computationally efficient variants.
    
  \end{minipage}
  \hfill
  \begin{minipage}[t]{0.45\textwidth}
    \vspace{0pt}
    \centering
    \includegraphics[width=\linewidth]{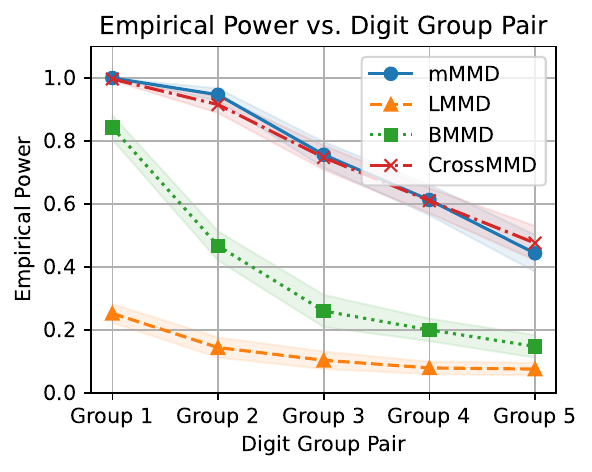}
    \caption{Empirical power comparison on MNIST.}
    \label{fig:mnist_asymptotic}
  \end{minipage}
\end{figure}
\vspace{-15pt}
\noindent

\section{Extensions}\label{sec:broad_scope}
In this section, we explore two broad applications of the \emph{mMMD}-test. In Section~\ref{sec:multi_kernel}, we propose combining multiple kernels with the $m\MMD$-test to potentially enhance power, and establish the null distribution of a corresponding Mahalanobis-distance-type test statistic. In Section~\ref{sec:general_mMMD}, we introduce a general family of MMD-based test statistics that interpolate between the quadratic-time MMD and the $m\MMD$ statistic. We study this family in more detail in Section~\ref{sec:Tngamma}.

\subsection{\texttt{mmMMD:} Multi-kernel test using $m\MMD$}\label{sec:multi_kernel}
We explore the broader applicability of the \emph{mMMD} test within two-sample testing, particularly as a complement to recent multi-kernel approaches for boosting power. As noted in Section~\ref{sec:experiments}, we used the widely adopted \emph{median heuristic} to select bandwidths for Gaussian and Laplace kernels. Though common, this heuristic lacks theoretical support and can underperform when the distributional difference occurs at scales far from the median pairwise distance. An alternative is to split the data and tune kernel parameters to maximize test power on a held-out set \citep{gretton2012optimal,liu2020learning}. More recent works improve robustness and power by aggregating over multiple kernels or bandwidths \citep{schrab2023mmd,schrab2022efficient,biggs2023mmd,chatterjee2025boosting}, enabling detection of both local and global differences. However, these methods still rely on the quadratic-time MMD and face significant computational costs due to the need for resampling or permutation to calibrate the test statistic.\\

\noindent
This limitation presents a promising opportunity to combine multiple kernels with the $m\MMD$ test introduced in \eqref{eq:mMMD_test}. To that end in the following we establish the asymptotic distribution of $m\MMD$ test with multiple kernels. Consider $\bm \cK_r = \{\sfK_1,\ldots,\sfK_r\}$ to be a collection of bounded positive definite kernels on $\cX = \R^d$ for some $d\geq 1$, such that each $\sfK\in \bm\cK_r$ is characteristic for $\cP(\cX)$. Let $T_{n,\ell}$ (recall \eqref{eq:Tn_alt}) be the (un-normalised) $m\MMD$ statistic computed using the kernel $\sfK_\ell$. With the notation in place we now establish the asymptotic null distribution of $\bm T_{n,r} = (T_{n,1},\ldots, T_{n,r})$.
\begin{theorem}\label{thm:multi_kernel_mMMD}
    Let $\Sigma_{n,r} = (\sigma_{n}(a,b))_{a,b=1}^{r}$ be the non-negative definite matrix with entries $\sigma_{n}(a,b) := \frac{1}{n^2}\sum_{i=2}^{n}\left(\frac{1}{i}\sum_{j=1}^{i-1}\sfH_a(\bZ_i,\bZ_j)\right)\left(\frac{1}{i}\sum_{j=1}^{i-1}\sfH_b(\bZ_i,\bZ_j)\right)$ where $\sfH_\ell$ is defined as in \eqref{eq:def_Hn_main} using the kernel $\sfK_\ell$ and $\bZ_i = (X_i, Y_i), 1\leq i\leq n$. Then under $\bH_0$ from \eqref{eq:H01}, $T_{n,r}^{\top}\Sigma_{n,r}^{-1}T_{n,r}\dto \chi_r^2$.
\end{theorem}
The proof of Theorem~\ref{thm:multi_kernel_mMMD} follows from the Cramér–Wold device and the Martingale Central Limit Theorem \citep[Theorem 3.2]{hall2014martingale}, along with the arguments in Sections~\ref{sec:proof_of_4thmmntbdd} and~\ref{sec:proofof_second_term}, and is therefore omitted. Theorem~\ref{thm:multi_kernel_mMMD} enables a test of $\bH_0$ by comparing the statistic to the appropriate quantile of the $\chi^2_r$ distribution. Importantly, this avoids the need for resampling or permutation, making the test simple to implement. For brevity, Theorem~\ref{thm:multi_kernel_mMMD} is stated under specific assumptions. A more general version, analogous to Theorem~\ref{thm:null_clt_general}, is possible but involves significantly more notation which are omitted.

We empirically assess the performance of the Multi-kernel mMMD (hereafter referred to as \texttt{mmMMD}) in Figure~\ref{fig:mmmmd_power_comparison}. In our experimental setup, we consider two distributions: $P = \mathcal{N}_{10}(\mathbf{0}, \boldsymbol{\Sigma})$ and $Q = \mathcal{N}_{10}(\mathbf{0}, 1.3\boldsymbol{\Sigma})$, where $\boldsymbol{\Sigma}$ denotes a covariance matrix with entries defined by $\Sigma_{ij} = 0.5^{|i-j|}$. This setting is adapted from \cite{chatterjee2025boosting}.

To evaluate statistical power, we consider three kernel configurations. For \texttt{mmMMD Gauss}, we use a combination of three Gaussian kernels with bandwidths $(1, 2, 4)\lambda_{\text{med}}$, where $\lambda_{\text{med}}$ is the median heuristic bandwidth. Similarly, for \texttt{mmMMD Laplace}, we employ three Laplace kernels with the same set of bandwidths. The third configuration, \texttt{mmMMD Mixed}, combines both Gaussian and Laplace kernels, each with the median bandwidth $\lambda_{\text{med}}$. For estimation of the empirical null distribution, we adopt the same Gaussian kernel configuration as used in \texttt{mmMMD Gauss}.

\begin{figure}[!t]
  \centering

  \begin{minipage}[b]{0.45\textwidth}
    \centering
    \includegraphics[width=1\linewidth]{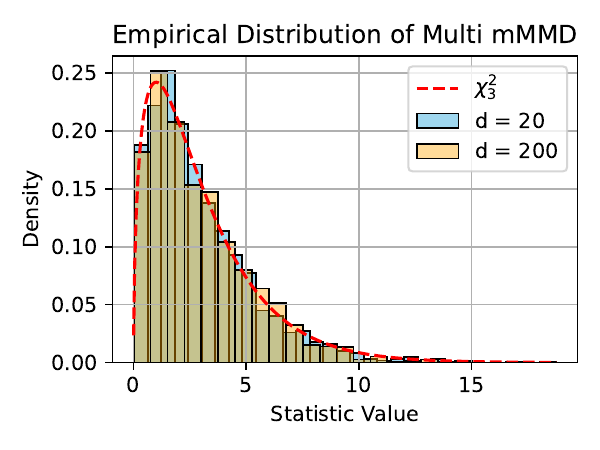}
  \end{minipage}
  \hfill
  \begin{minipage}[b]{0.45\textwidth}
    \centering
    \includegraphics[width=0.95\linewidth]{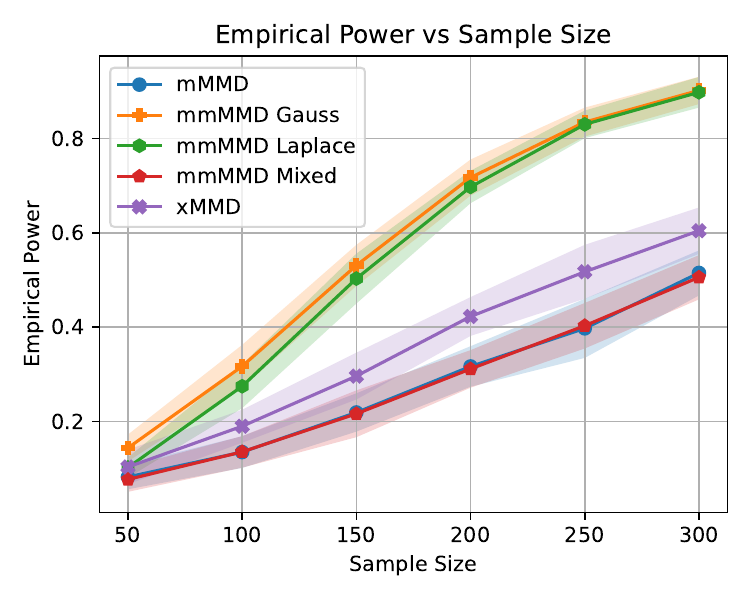}
  \end{minipage}

  \caption{(Left) Empirical null distribution of \texttt{mmMMD} and (Right) Empirical power of \texttt{mmMMD}.}
  \label{fig:mmmmd_power_comparison}
\end{figure}

Figure~\ref{fig:mmmmd_power_comparison} validates the empirical null distribution derived in Theorem~\ref{thm:multi_kernel_mMMD}, and further suggests that a dimension-agnostic result similar to Theorem~\ref{thm:null_clt_general} may also hold for \texttt{mmMMD}. Moreover, the power comparison indicates that combining multiple kernels can lead to significant improvements in test power. However, it also suggests that arbitrary kernel combinations do not necessarily yield such improvements, as evidenced by the performance of \texttt{mmMMD Mixed} in the power plot. This observation highlights the need for a more principled strategy for selecting and combining kernels, pointing to an interesting future research direction in the study of optimal kernel combinations.

\subsection{Generalizing the $m\MMD$ Statistic}\label{sec:general_mMMD}

In this section, we briefly introduce a generalization of the $m\MMD$ statistic, defining a family of statistics that interpolate between the original quadratic-time MMD statistic from \cite{gretton2012kernel} and our proposed linear-time $m\MMD$. Recalling the alternate form of $T_n$ from \eqref{eq:Tn_alt}, we define the generalized statistic as
\begin{align*}
T_{n,\gamma} = \frac{1}{n^{2-\gamma}}\sum_{i=2}^{n}\frac{1}{i^{\gamma}}\sum_{j=1}^{i-1}\sfH_n(\bZ_i,\bZ_j), \quad \text{for } 0 \leq \gamma \leq 1,
\end{align*}
where $\sfH_n$ is defined in \eqref{eq:def_Hn_main}. This family recovers the quadratic-time statistic from \cite{gretton2012kernel} when $\gamma = 0$ and the $m\MMD$ statistic when $\gamma = 1$, thus forming a continuous interpolation between the two.\\

\noindent
We study this family in detail in Section~\ref{sec:Tngamma}, where we establish the asymptotic null distribution under $\bH_0$ as standard Gaussian (Theorem~\ref{thm:berry_essen_H0}) and derive conditions for consistency under $\bH_1$ (Theorem~\ref{thm:general_consistency}). These results generalize the corresponding theorems for the special case $\gamma = 1$ (the $m\MMD$-test). Moreover, we show that a broad subclass of this family is minimax optimal for detecting distributional differences when the $L_2$ distance between the densities vanishes at the minimax rate \citep{li2024optimality}.

\begin{remark}\label{remark:non_minimax}
    Although $T_{n,\gamma}$ defines a general family of statistics encompassing the $m\MMD$ test, a key limitation for $\gamma < 1$ is that our proof technique does not establish asymptotic Gaussianity under the null when the kernel is fixed (i.e., does not vary with $n$), unlike the $\gamma=1$ case. In particular, for $\gamma = 0$, corresponding to the quadratic-time $\MMD$ test, the null distribution is known to be non-Gaussian in this setting. Similar behavior might hold for $\gamma < 1$ cases. On the other hand, in Section~\ref{sec:minimax}, we show that $T_{n,\gamma}$ yields a minimax rate-optimal test (similar to \cite{li2024optimality}) whenever $\gamma \leq 1/2$. However, this result does not apply to the $m\MMD$ test. In fact, our current analysis does not establish that the $m\MMD$ test is consistent when the $L_2$ distance between the underlying densities decays at the minimax rate. The current proof technique appears to be inadequate for verifying the validity of the $m\MMD$ test with respect to Type-I error control. Nonetheless, it can be shown that, up to logarithmic factors, the consistency part of the proof still holds. Whether or not $\gamma > 1/2$ is minimax optimal remains an open problem. Nonetheless, given the close connection to the cross-MMD test from \citet{shekhar2022permutation}, which is provably minimax optimal, and the empirically similar performance of the $m\MMD$ test, a promising future direction is to characterize the class of alternatives for which the $m\MMD$ test is rate-optimal. 
\end{remark}

\section{Conclusion and Future Work}
In this work, we proposed a novel variant of the kernel MMD test, called the $m$MMD test, by identifying a modified form of the squared MMD estimator as a martingale. We showed that this formulation admits a standard Gaussian distribution under the null, allowing for a resampling-free, computationally efficient, and consistent kernel test for the two-sample problem. Experiments demonstrate that our test achieves power close to the quadratic-time MMD test and often matches or surpasses the performance of other computationally efficient alternatives.

The established equivalence between kernel-based and distance-based tests \citep{sejdinovic2013equivalence} suggests that our method extends naturally to distance-based settings. Furthermore, building on the connection between kernel-based two-sample and independence testing \citep{gretton2007kernel}, an interesting future direction is to adapt our approach for computationally efficient independence testing. Additional computational gains might also be achieved by integrating our method with recent advances in kernel approximation techniques \citep{choi2024computational, chatalic2025efficient, domingo2023compress}.

\acks{We thank Ilmun Kim for some useful technical discussions.}

\bibliography{references}

\newpage

\appendix



\part*{Supplementary Materials} 

\begingroup
\setcounter{tocdepth}{2} 
\tableofcontents
\endgroup


\section{A general class of Martingale based MMD two-sample test}\label{sec:Tngamma}
In this section, we introduce a general class of statistics using the squared Maximum Mean Discrepancy (MMD) (recall \eqref{eq:def_MMD}). The \emph{mMMD} estimator defined in \eqref{eq:Tn_alt} is a special case within this broader family. Moreover, we show that the original quadratic-time MMD estimator proposed by \citet{gretton2012kernel} also belongs to this class.\\

\noindent
In Section~\ref{sec:null_dist_gen}, we show that in the most general setting, where the kernel $\mathsf{K}_n$, the null distribution $P_n (= Q_n)$, and the sample space $\cX_n = \R^{d_n}$ may all vary with the sample size $n$, the entire family of test statistics converges in distribution to a standard Gaussian under the null. Then, in Section~\ref{sec:consistency_gen}, we provide general conditions under which this family of tests achieves uniform consistency for the two-sample testing problem (see Eq.~\eqref{eq:H01}). Finally, in Section~\ref{sec:minimax}, we focus on a specific subclass of tests based on the Gaussian kernel, and show that it can consistently detect differences between $P_n$ and $Q_n$ when they admit densities and the $L_2$ distance between those densities converges to zero at the minimax optimal rate.\\

\noindent
We begin by recalling and setting up the notation. Let $n\geq 1$ and take the sample space $\cX_n = \R^{d_n}$ (we allow the dimension to vary with $n$). Suppose we observe independent samples $\sX_n := \{X_1,\ldots, X_n\}\sim P_n$ and $\sY_n := \{Y_1,\ldots, Y_n\}\sim Q_n$. Take $\sfK_n$ to be a positive definite kernel with associated RKHS $\cK_n$ satisfying Assumption~\ref{assumption:K} and moreover suppose $P_n,Q_n\in \cM_{\sfK_n}^{1/2}$. For $\bZ_i = (X_i, Y_i), 1\leq i\leq n$ define,
\begin{align}\label{eq:def_H_n}
  \sfH_n\left(\bZ_i, \bZ_j\right) := \sfK_n(X_i, X_j) - \sfK_n(X_i, Y_j) - \sfK_n(X_j, Y_i) + \sfK_n(Y_i, Y_j).
\end{align}
With the notation in place we now define the family of estimators,
\begin{align}\label{eq:def_T_n_gamma}
  T_{n,\gamma} := \frac{1}{n^{2-\gamma}}\sum_{i=2}^{n}\frac{1}{i^\gamma}\sum_{j=1}^{i-1}\sfH_n\left(\bZ_i,\bZ_j\right),\text{ for all }0\leq \gamma\leq 1.
\end{align}

\begin{remark}
Recall that the standard quadratic-time estimator of $\mathrm{MMD}^2[\sfK_n, P_n, Q_n]$, introduced by \citet{gretton2012kernel}, is given by
$
\frac{1}{n^2} \sum_{i \neq j} \mathsf{H}_n(\mathbf{Z}_i, \mathbf{Z}_j).
$
By definition, and using the symmetry of the kernel function $\mathsf{H}_n$, this expression is equivalent to $2T_{n,0}$. Hence, the classical quadratic-time MMD estimator is a special case of our proposed family.
Moreover, the $m\mathrm{MMD}$ statistic introduced in \eqref{eq:def_eta_n} (recall the equivalent form from \eqref{eq:Tn_alt}) can be recovered by setting $\gamma = 1$ in our framework. Specifically, observe that $T_{n,1}$ corresponds exactly to the $m\mathrm{MMD}$ statistic.
\end{remark}

\subsection{Asymptotic Null Distribution of $T_{n,\gamma}$}\label{sec:null_dist_gen}

In this section, we establish the asymptotic null distribution of $T_{n,\gamma}$, that is, its distribution under the null hypothesis $\bH_0$ where $P_n = Q_n$. To this end, we first define the centered kernel $\bar\sfK_n$ as,
\begin{align}\label{eq:defbarK}
  \bar\sfK_n(x,y) := \left\langle\sfK_n(x,\cdot) - \nu_{P_n}, \sfK_n(y,\cdot) - \nu_{P_n}\right\rangle_{\cK_n}.
\end{align}

\noindent
Next, we define the terms,
    \begin{align}\label{eq:tnalpha}
        t_{n,\gamma,1} = 
        \begin{cases}
        n^{-2} & \text{ if }0\leq \gamma<\frac{1}{2}\\
        \frac{\log n}{n^2} & \text{ if }\gamma = \frac{1}{2}\\
        n^{-4(1-\gamma)} & \text{ if }\frac{1}{2}<\gamma< 1\\
        \frac{1}{(\log n)^2} & \text{ if }\gamma = 1,
        \end{cases}
        \text{ and }
        t_{n,\gamma,2} = 
        \begin{cases}
        n^{-1} & \text{ if }0\leq \gamma<\frac{3}{4}\\
        \frac{\log n}{n} & \text{ if }\gamma = \frac{3}{4}\\
        n^{-4(1-\gamma)} & \text{ if }\frac{3}{4}<\gamma< 1\\
        \frac{1}{(\log n)^2} & \text{ if }\gamma = 1.
        \end{cases}
    \end{align}

\noindent
With the notation in place we now establish a Berry-Essen type distributional convergence results for the family $T_{n,\gamma}$ when $P_n = Q_n$.

\begin{theorem}\label{thm:berry_essen_H0}
  Generate $X_1,X_2,X_3$ independently from $P_n$ and suppose $\E\left[\bsfK_n(X_1,X_2)^4\right]<\infty$. Then define, 
    \begin{align}\label{eq:def_Rn1}
        R_{n,\gamma,1} = t_{n,\gamma,1}\frac{\E\left[\bsfK_n(X_1, X_2)^4\right]}{\E\left[\bsfK_n(X_1, X_2)^2\right]^2} + t_{n,\gamma,2}\left(\frac{\E\left[\bsfK_n(X_1, X_2)^2\bsfK_n(X_1, X_3)^2\right]}{\E\left[\bsfK_n(X_1, X_2)^2\right]^2} + 1\right) ,
    \end{align}
    where $t_{n,\gamma,1}$ and $t_{n,\gamma,2}$ are defined in \eqref{eq:tnalpha} and also let,
    \begin{align}\label{eq:def_Rn2}
        R_{n,\gamma,2} = 
        \begin{cases}
        \frac{\E\left[\E\left[\bsfK(X_1, X_2)\bsfK(X_1,X_3)\mid X_2, X_3\right]^2\right]}{\E\left[\bsfK_n(X_1, X_2)^2\right]^2} & \text{ if }0\leq \gamma <1\\
        \frac{1}{\log n}& \text{ if }\gamma = 1.
        \end{cases}
    \end{align}
    Then under $\bm H_0$,
    \begin{align*}
        \sup_{x\in \R}\left|\P\left(n^{2-\gamma}\frac{T_{n,\gamma}}{\sigma_{n,\gamma}}\leq x\right) - \Phi(x)\right|\lesssim (R_{n,\gamma,1} + R_{n,\gamma,2})^{1/5} ,
    \end{align*}
    where $\sigma_{n,\gamma}^2 = \sum_{i=2}^{n}\left(\frac{1}{i^\gamma}\sum_{j=1}^{i-1}\sfH_n(\bZ_i, \bZ_j)\right)^2$ and $\Phi(\cdot)$ is the CDF of $\mathrm{N}(0,1)$.
\end{theorem}

We note that the terms $R_{n,\gamma,1}$ and $R_{n,\gamma,2}$ are well-defined, as established in Lemma~\ref{lemma:E_barK_sq_positive}. The proof of Theorem~\ref{thm:berry_essen_H0} is presented in Section~\ref{sec:proofof_berry_essen}. The key idea is to recognize $T_{n,\gamma}$ as a sum of martingale differences and then apply the Berry–Esseen bounds for self-normalized martingales from \cite{fan2018berry}. To simplify the above result and to obtain conditions under which the asymptotic null distribution of $T_{n,\gamma}$ is standard Gaussian, we introduce the following assumptions.
\begin{assumption}\label{assumption:kernelconvg}
    For the centered kernel $\bar\sfK_n$ and the sequence of distributions $\{P_n: n\geq 1\}$, we assume that
    \begin{align*}
        \lim_{n\ra\infty}\frac{\E\left[\E_{P_n}\left[\bsfK_n(X_1, X_2)\bsfK_n(X_1, X_3)\mid X_2,X_3\right]^2\right]}{\E\left[\bsfK_n(X_1, X_2)^2\right]^2} = 0,
    \end{align*}
    where $X_1, X_2, X_3$ are generated independently from $P_n$.
\end{assumption}

\begin{remark}
    The condition in Assumption~\ref{assumption:kernelconvg} indeed holds true with $L = 0$ when $\sfK_n(x,y) = \exp\left(-\nu_n\left\|x-y\right\|_2^2\right)$ is the Gaussian kernel with bandwidth $\nu_n\rightarrow\infty$. This follows from condition $(9)$ of \cite{li2024optimality}. When the kernel $\sfK_n$ and distribution $P_n$ do not change with $n$, then it is immediately clear that the limit $L$ exists but it might not be $0$.
\end{remark}

\begin{assumption}\label{assumption:4thand22moment}  For a kernel $\sfK_n$ and the sequence of distributions $\{P_n: n\geq 1\}$, assume that there exists some $\gamma \in [0,1]$ such that
    \begin{align*}
        \lim_{n\rightarrow \infty}t_{n,\gamma,1}\frac{\E\left[\bsfK_n(X_1, X_2)^4\right]}{\E\left[\bsfK_n(X_1, X_2)^2\right]^2} = 0\text{ and }\lim_{n\rightarrow \infty}t_{n,\gamma,2}\frac{\E\left[\bsfK_n(X_1, X_2)^2\bsfK_n(X_1, X_3)^2\right]}{\E\left[\bsfK_n(X_1, X_2)^2\right]^2} = 0 ,
    \end{align*}
    where $X_1,X_2,X_3$ are generated independently from $P_n$.
\end{assumption}
One can easily verify that this condition is satisfied whenever the kernel $\mathsf{K} = \mathsf{K}_n$ and the null distribution $P = P_n$ remain fixed as $n$ increases. Under the assumptions stated above, we now obtain the following corollary to Theorem~\ref{thm:berry_essen_H0}, which establishes the convergence of the null distribution of $T_{n,\gamma}$ to the standard Gaussian.

\begin{corollary}\label{cor:NullCLT}
  Let $\E\left[\bsfK_n(X_1, X_2)^4\right]<\infty$ and consider some $\gamma\in [0,1]$ such that Assumption~\ref{assumption:4thand22moment} holds. Furthermore, adopt Assumption~\ref{assumption:kernelconvg} if $\gamma\in [0,1)$. Then under $\bm H_0$,
    \begin{align*}
        \sup_{x\in \R}\left|\P\left(n^{2-\gamma}\frac{T_{n,\gamma}}{\sigma_{n,\gamma}}\leq x\right) - \Phi(x)\right|\ra 0.
    \end{align*}
\end{corollary}
Note that when $\gamma=1$ the Assumption~\ref{assumption:kernelconvg} is not required. Corollary~\ref{cor:NullCLT} follows immediately from Theorem~\ref{thm:berry_essen_H0}, and its proof is therefore omitted. With the result from Corollary~\ref{cor:NullCLT} we can now define the test $\phi_{n,\gamma}$ for $\bm H_0: P_n = Q_n$ as,
\begin{align}\label{eq:def_test}
    \phi_{n,\gamma} = \one\left\{\frac{\sum_{i=2}^{n}\frac{1}{i^{\gamma}}\sum_{j=1}^{i-1}\sfH(\bZ_i, \bZ_j)}{\sqrt{\sum_{i=2}^{n}\left(\frac{1}{i^\gamma}\sum_{j=1}^{i-1}\sfH(\bZ_i, \bZ_j)\right)^2}}>z_{\gamma}\right\}\text{ for all }0\leq \gamma\leq 1.
\end{align}
Then under the assumptions of Corollary~\ref{cor:NullCLT}, the test $\phi_{n,\gamma}$ is asymptotically valid at level $\gamma$.

\subsection{Consistency against fixed alternatives}\label{sec:consistency_gen}
In this section we establish the consistency of the test $\phi_{n,\gamma}$ for any $0\leq \gamma\leq 1$. Recall the general setting where we allow the kernel $\sfK_n$ and the sample space $\cX_n = \R^{d_n}$ to change with $n$ and also consider $(P_n, Q_n)\in \cP_n\subseteq \cM_{\sfK_n}^{1/2}\times \cM_{\sfK_n}^{1/2}$ which can also change with $n$. In this setting we identify sufficient conditions under which the test $\phi_{n,\gamma}$ is consistent uniformly over $\cP_n$. Let $\delta_n = \MMD[P_n, Q_n, \sfK_n]$ which can approach $0$ in the limit.

\begin{theorem}\label{thm:general_consistency}
    Define $s_{n,\gamma} = \sum_{i=2}^{n}(i-1)/i^\gamma$. If $\delta_n>0$ and,
    \small
    \begin{align}\label{eq:genconsistent_assumption}
        \lim_{n\ra\infty} \sup_{(P_n, Q_n)\in \cP_n} \frac{\E\left[\frac{1}{s_{n,\gamma}^2}\sum_{i=2}^n\left(\frac{1}{i^\gamma}\sum_{j=1}^{i-1}\sfH(\bZ_i,\bZ_j)\right)^2\right]}{\delta_n^4} + \frac{\Var\left(\frac{1}{s_{n,\gamma}}\sum_{i=2}^{n}\frac{1}{i^\gamma}\sum_{j=1}^{i-1}\sfH(\bZ_i, \bZ_j)\right)}{\delta_n^4} = 0,
    \end{align}
    \normalsize
    then $\lim_{n\ra\infty}\sup_{(P_n, Q_n)\in \cP_n}\E\left[1-\phi_{n,\gamma}\right] = 0$.
\end{theorem}

The proof of Theorem~\ref{thm:general_consistency} is similar to the proof of Theorem 8 from \cite{shekhar2022permutation}; however, for the sake of completeness, we provide a proof in Appendix~\ref{sec:proofofgeneralconsistentcy}. 

\subsection{Minimax optimality against smooth local alternatives}\label{sec:minimax}
In this section, we apply the above general consistency result to distributions $P_n$ and $Q_n$ which admit densities $p_n$ and $q_n$ lying in the order $\beta$ Sobolev ball for some $\beta>0$ denoted as
\small
\begin{align*}
    \cW^{\beta,2}(M) = \left\{f: \cX\ra \R \mid f\text{ is a.s.\ continuous, and }\int (1+\omega^2)^{\beta/2}\left\|\cF(f)(\omega)\right\|^2\d\omega<M<\infty\right\}.
\end{align*}
\normalsize
Take $\cX = \R^d$, that is the sample space (in particular the dimension) does not change with $n$. Formally the null and alternative class of distributions are defined as
\begin{align*}
    \cP_n^{(0)} 
    & := \left\{P_n \text{ with density }p_n: p_n\in \cW^{\beta,2}(M)\right\}\\
    \cP_n^{(1)}
    & := \left\{(P_n, Q_n) \text{ with densities }(p_n, q_n)\in \cW^{\beta,2}(M)\text{ respectively}: \left\|p_n-q_n\right\|_2\geq \Delta_n\right\}
\end{align*}
for some sequence $\Delta_n$ decaying to $0$. In particular, under $\bm H_0$, $P_n=Q_n\in \cP_n^{(0)}$ and under the alternative we assume $(P_n, Q_n)\in \cP_{n}^{(1)}$. The following result shows that the test $\phi_{n,\gamma}$ from \eqref{eq:def_test} with certain choices of $\gamma$ and the Gaussian kernel (with a suitably chosen scale parameter) is minimax rate optimal for the above class of local alternatives.

\begin{theorem}\label{thm:optimality}
    Let $\Delta_n$ be a sequence such that $\lim_{n\ra\infty}\Delta_n n^{\frac{2\beta}{d+4\beta}} = \infty$. Then for the test $\phi_{n,\gamma}$ from \eqref{eq:def_test} with the Gaussian kernel $\sfK_n(x,y) = \exp(-\nu_n\left\|x-y\right\|^2)$ having bandwidth $\nu_n = n^{4/(d+4\beta)}$ we have
    \begin{align*}
        \lim_{n\ra\infty} \P_{(P_n = Q_n)\in \cP_n^{(0)}}\left(\phi_{n,\gamma} = 1\right) \leq \delta\text{ and }\lim_{n\ra\infty}\sup_{(P_n, Q_n)\in \cP_n^{(1)}}\P(\phi_{n,\gamma} = 1) = 1,
    \end{align*}
    whenever $0\leq \gamma< \frac{d+8\beta}{2d+8\beta.}$.
\end{theorem}
The proof of Theorem~\ref{thm:optimality} is provided in Section~\ref{sec:proof_minimax} and it follows by verifying the conditions of Theorem~\ref{thm:general_consistency}.

\begin{remark}
    We note that Theorem~\ref{thm:optimality} establishes only one direction of the minimax optimality claim for $\phi_{n,\gamma}$. The complementary lower bound is provided by \cite{li2024optimality}, who show that if $\Delta_n n^{\frac{2\beta}{d + 4\beta}} = O(1)$, then there exists some $\alpha \in (0,1)$ such that no asymptotically level-$\alpha$ test can achieve power tending to one. This implies that the sequence $\Delta_n$ in Theorem~\ref{thm:optimality} is indeed the smallest separation between $P_n$ and $Q_n$ (in terms of $L_2$ distance between densities) that any test can reliably detect. Consequently, $\phi_{n,\gamma}$ is minimax optimal for $0 \leq \gamma < \frac{d + 8\beta}{2d + 8\beta}$. The status of $\gamma$ above this range is presently unknown.
\end{remark}


\section{Additional Experiments}
In this section we include additional experiments, including validity of $m\MMD$ tests under the null, comparative performance of $T_{n,\gamma}$ for different values of $\gamma$ and comparative performance of the tests from Section \ref{sec:experiments} in the minimax regime.

\subsection{Comparing performance of $T_{n,\gamma}$ for different $\gamma$}

\begin{figure}[H]
  \centering

  \begin{minipage}[t]{0.55\textwidth}
  
    In this section, we compare the performance of tests based on $T_{n,\gamma}$ for 
$\gamma \in \{0, 0.25, 0.5, 0.75, 1\}$. We recall the experimental setting from Section~\ref{sec:real_data} and apply the $T_{n,\gamma}$-based tests to distinguish between the distributions $P$ and $Q$. To ensure a fair comparison across all values of $\gamma$, we use a permutation-based calibration procedure for every test. The entire experiment is repeated $100$ times, and the resulting empirical performance is summarized in Figure~\ref{fig:mnist_asymptotic_gamma}.\\

As noted previously (see Section~\ref{sec:real_data}), the testing problem becomes increasingly difficult as the group index grows. Consistent with this, the power of the $T_{n,\gamma}$-based tests decreases for all values of $\gamma$. Moreover, the 
  \end{minipage}
  \hfill
  \begin{minipage}[t]{0.4\textwidth}
    \vspace{0pt}
    \centering
    \includegraphics[width=\linewidth]{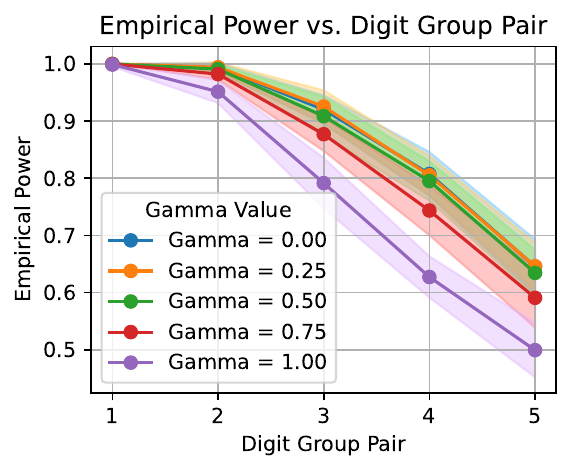}
    \caption{Empirical power comparison.}
    \label{fig:mnist_asymptotic_gamma}
  \end{minipage}
\end{figure}
\vspace{-13pt}
\noindent
rate of decline becomes steeper as $\gamma$ increases. This behavior matches the intuition that $m\MMD$-based tests (corresponding to $\gamma = 1$) trade off statistical power for computational efficiency relative to $\MMD$ (corresponding to $\gamma = 0$), and that the family $\{T_{n,\gamma}\}_{\gamma \in [0,1]}$ provides a continuous interpolation 
between these two extremes.

\subsection{Type-I validity of $m\MMD$}
In this section, we revisit the experimental setup from Section~\ref{sec:power_sims}, fix $d = 10$, and examine Type-I error control for the quadratic-time $\mathrm{MMD}$, 
$m\mathrm{MMD}$, $x\mathrm{MMD}$, the linear-time, and the block-based tests in both small- and large-sample regimes.

\begin{figure}[htbp]
    \centering
    \begin{minipage}[b]{0.48\textwidth}
        \centering
        \includegraphics[width=\textwidth]{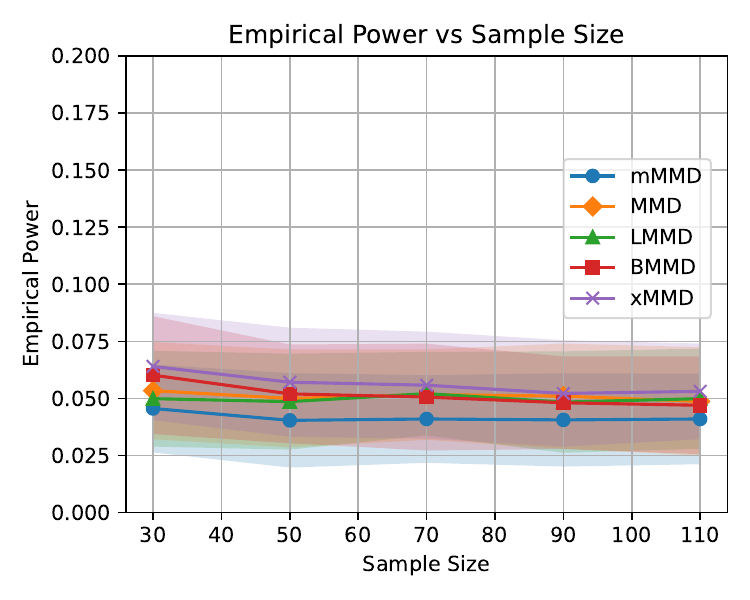}
        \caption*{(a) Type-I error control at small samples.}
    \end{minipage}
    \hfill
    \begin{minipage}[b]{0.48\textwidth}
        \centering
        \includegraphics[width=\textwidth]{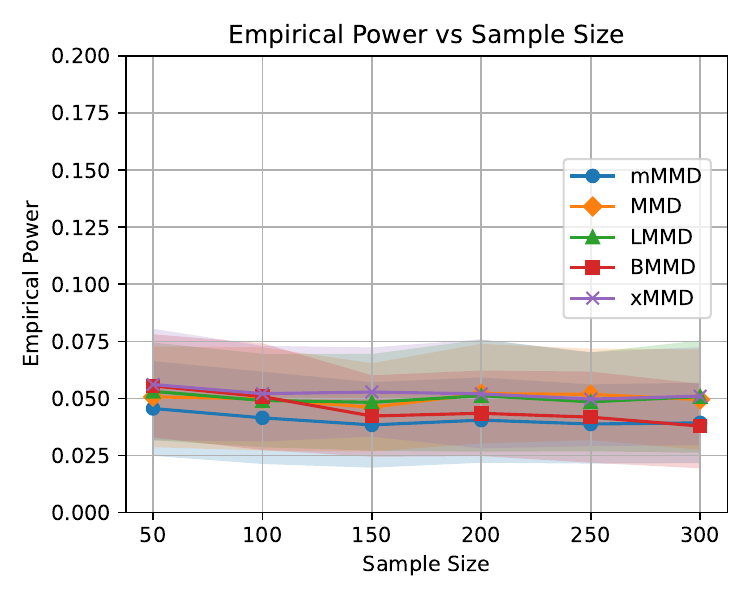}
        \caption*{(b) Type-I error control at large samples.}
    \end{minipage}

    \caption{Type-I error control for the original quadratic time \texttt{MMD}, \texttt{mMMD, xMMD} and the linear time \texttt{LMMD} and block-based \texttt{BMMD}.}
    \label{fig:type-I}
\end{figure}

Figure~\ref{fig:type-I} illustrates the Type-I error control of all the aforementioned tests in both small- and large-sample settings. While all methods maintain the nominal Type-I error level $\alpha = 0.05$ in large-sample scenarios, in the small-sample setting with $n = 10$, the $x\mathrm{MMD}$ and block-based test (\texttt{BMMD}) fail. This is because both rely on sample splitting and block-based approximations, which are less accurate in such small-sample regimes.

\subsection{Performance in minimax regime}\label{sec:minimax_empirical}
Building on Remark~\ref{remark:non_minimax} and Theorem~\ref{thm:optimality}, we observe that our current proof technique does not establish the rate-optimality of the $m\mathrm{MMD}$ test. Nevertheless, the experiments in Section~\ref{sec:experiments} indicate that, empirically, the $m\mathrm{MMD}$ test performs on par with the $x\mathrm{MMD}$ test of \cite{shekhar2022permutation}, 
which is known to be minimax optimal. To determine whether the absence of a rate-optimality guarantee for the $m\mathrm{MMD}$ test reflects a limitation of our analytical approach or an inherent limitation of the test itself, we design in this section a set of experiments that closely mirror the conditions of Theorem~\ref{thm:optimality}. In particular we take
\begin{align*}
    p_n  =\mathrm N(0,1) \text{ and }q_n = \mathrm N\left(2\sqrt{-\log\left(1-\Delta_n^2\sqrt{\pi}\right)}, 1\right).
\end{align*}
With these choices, it immediately follows that $\|p_n - q_n\|_2 = \Delta_n$. In Figure~\ref{fig:minimax_empirical}, we present the performance of the quadratic-time $\mathrm{MMD}$, $m\mathrm{MMD}$, and $x\mathrm{MMD}$ tests, as well as the linear-time and block-based $\mathrm{MMD}$ tests, in detecting the difference between $p_n$ and $q_n$. All tests use a Gaussian kernel with bandwidth $\nu_n = n^{4/5}$ (see the choice of bandwidth in Theorem~\ref{thm:optimality}), we vary $n\in \{100, 200, 300, 400, 500\}$ and we choose $\Delta_n = n^{-1/5}, n^{-1/4}$. Note that this choice of $\Delta_n$ satisfies the constraint from Theorem \ref{thm:optimality} with $\beta=1$.

\begin{figure}[htbp]
    \centering
    \begin{minipage}[b]{0.32\textwidth}
        \centering
        \includegraphics[width=\textwidth]{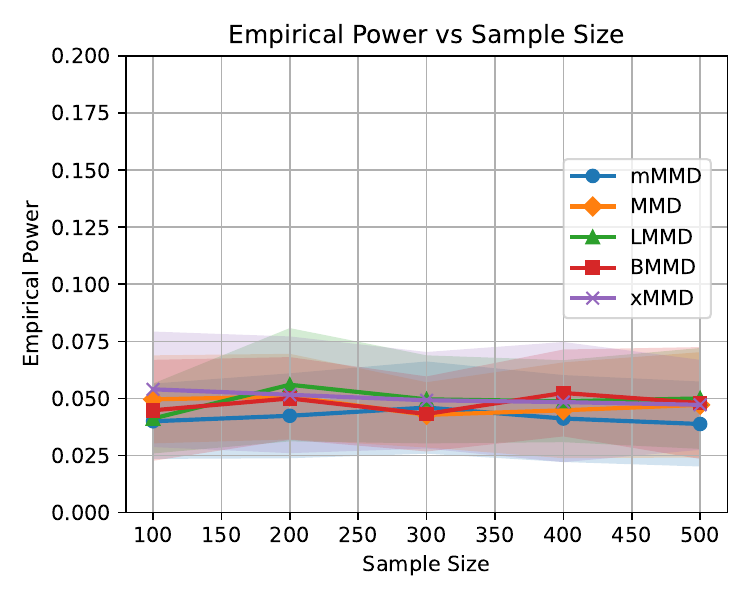}
        \caption*{(a) Type-I Error}
    \end{minipage}
    \hfill
    \begin{minipage}[b]{0.32\textwidth}
        \centering
        \includegraphics[width=\textwidth]{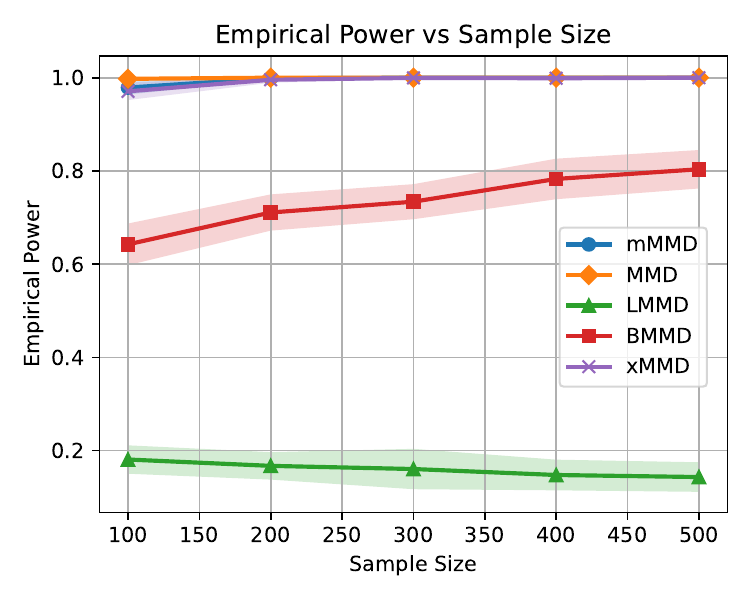}
        \caption*{(b) Power with $\Delta_n = n^{-\frac{1}{5}}$}
    \end{minipage}
    \hfill
    \begin{minipage}[b]{0.32\textwidth}
        \centering
        \includegraphics[width=\textwidth]{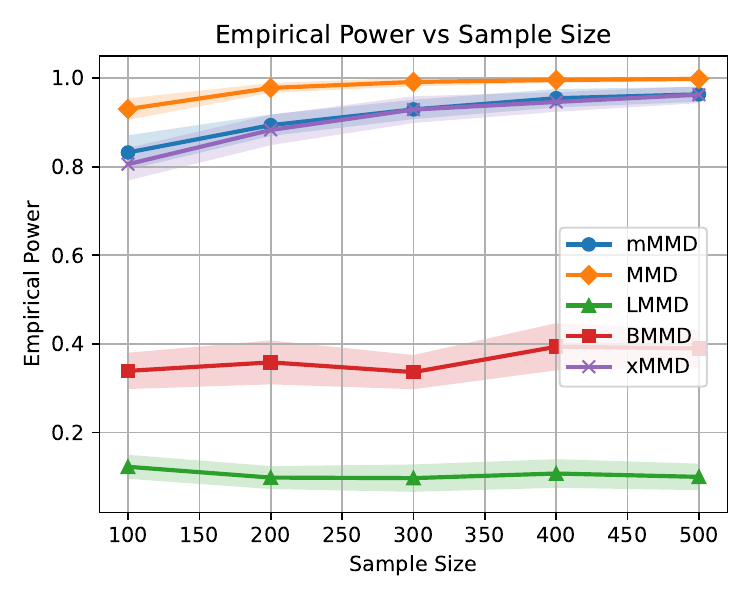}
        \caption*{(c) Power wth $\Delta_n = n^{-\frac{1}{4}}$}
    \end{minipage}

    \caption{Performance of $\MMD,m\MMD$, $x\MMD$ \citep{shekhar2022permutation} and the linear time \texttt{LMMD} and the block-based \texttt{BMMD} in the setting from Section \ref{sec:minimax_empirical}.}
    \label{fig:minimax_empirical}
\end{figure}

Figure~\ref{fig:minimax_empirical} suggests that, in the minimax regime, the $m\mathrm{MMD}$ test performs comparably to both the $x\mathrm{MMD}$ and the quadratic-time $\mathrm{MMD}$ tests in terms of type-I error control 
and statistical power. This indicates that the limitation from Remark \ref{remark:non_minimax} may arise from the proof technique rather than the test itself, and it is possible that $m\mathrm{MMD}$ also achieves rate-optimality.

\section{Proof of results from Section~\ref{sec:Tngamma}}
In this section we collected the deferred proofs of results from Section~\ref{sec:Tngamma}. 

\subsection{Proof of Theorem~\ref{thm:berry_essen_H0}}\label{sec:proofof_berry_essen}
The proof proceeds by recognizing \( T_{n,\gamma} \) as a martingale and applying the Berry--Esseen bounds for self-normalized martingales from \cite{fan2018berry}. To that end, we begin by introducing the necessary notation. Define
\begin{align}\label{eq:def_c_alpha}
    c_{\gamma} := 
    \begin{cases}
        \frac{2}{1-\gamma} & \text{ if }0\leq \gamma<1\\
        4 & \text{ if }\gamma = 1.
    \end{cases}
\end{align}
Take $M_n = \E\left[\bsfK_n(X_1, X_2)^2\right]$ where $X_1,X_2$ are generated independently from $P_n$. To identify $T_{n,\gamma}$ as a martingale, define $W_{n,1,\gamma} = 0, 0\leq \gamma\leq 1$ and
\begin{align*}
    W_{n,i,\gamma} = 
    \begin{cases}
        \frac{1}{n^{1-\gamma}\sqrt{c_\gamma M_n}i^\gamma}\sum_{j=1}^{i-1}\sfH_n(\bZ_i, \bZ_j) & \text{ if }0\leq \gamma < 1\\
        \frac{1}{\sqrt{\log n}\sqrt{c_1 M_n}i}\sum_{j=1}^{i-1}\sfH_n(\bZ_i, \bZ_j) & \text{ if }\gamma = 1.
    \end{cases}
    \text{ for all }2\leq i\leq n,
\end{align*}
where we take $\bZ_i = (X_i, Y_i), 1\leq i\leq n$ and $\sfH_n$ is defined in \eqref{eq:def_H_n}. Moreover, consider the sigma algebras $\cF_{n,1} = \cF_1 = \left\{\emptyset, \Omega\right\} $ and $\cF_{n,i}:=\cF_i = \sigma\left(\bZ_1,\ldots, \bZ_i\right)$ for all $2\leq i\leq n$. Finally, let
\begin{align*}
  S_{n,i,\gamma} = \sum_{j=1}^{i}W_{n,j,\gamma} \text{ for all }0\leq \gamma\leq 1.
\end{align*}
It is easy to see that under the null, for $n\geq 1$, $\{(S_{n,i,\gamma}, \cF_{n,i}): 1\leq i\leq n\}$ forms a zero-mean, square integrable martingale. Finally, notice that,
\begin{align*}
    n^{2-\gamma}\frac{T_{n,\gamma}}{\sigma_{n,\gamma}} = \frac{\sum_{i=2}^{n}\frac{1}{i^{\gamma}}\sum_{j=1}^{i-1}\sfH(\bZ_i, \bZ_j)}{\sqrt{\sum_{i=2}^{n}\left(\frac{1}{i^\gamma}\sum_{j=1}^{i-1}\sfH(\bZ_i, \bZ_j)\right)^2}} = \frac{S_{n,n,\gamma}}{\sqrt{\sum_{i=2}^{n}W_{n,i,\gamma}^2}}.
\end{align*}
Then by Theorem 2.1 from \cite{fan2018berry} (the required moment condition is satisfied by following computations similar to \eqref{eq:4thmomentbdd}, \eqref{eq:Hn22bddKn22} and the moment assumption on $\bsfK_n$) we get,
\begin{align*}
    \sup_{x\in \R}\left|\P\left(n^{2-\gamma}\frac{T_{n,\gamma}}{\sigma_{n,\gamma}}\leq x\right) - \Phi(x)\right|\lesssim N_{n,\gamma}^{1/5},
\end{align*}
where
\begin{align}\label{eq:N_n_alpha}
    N_{n,\gamma} = \sum_{i=2}^{n}\E\left[W_{n,i,\gamma}^4\right] + \E\left[\left|\sum_{i=2}^{n}\E\left[W_{n,i,\gamma}^2\mid \cF_{n,i-1}\right] - 1\right|^2\right].
\end{align}
To complete the proof it is now enough to upper bound $N_{n,\gamma}$. We begin by a simple upper bound on $N_{n,\gamma}$. Note that
\small
\begin{align}\label{eq:replace_W_EWcF}
    \E\left[\left(\sum_{i=2}^{n}W_{n,i,\gamma}^2 - \E\left[W_{n,i,\gamma}^2\mid\cF_{n,i-1}\right]\right)^2\right] = \sum_{i=2}^{n}\E\left[\left(W_{n,i,\gamma}^2 - \E\left[W_{n,i,\gamma}^2\right]\right)^2\right]\lesssim\sum_{i=2}^{n}\E\left[W_{n,i,\gamma}^4\right].
\end{align}
\normalsize
Then recalling definition of $N_{n,\gamma}$ from \eqref{eq:N_n_alpha} and by \eqref{eq:replace_W_EWcF} we get,
\begin{align}\label{eq:bdd_N_alpha}
    N_{n,\gamma}\lesssim \sum_{i=2}^{n}\E\left[W_{n,i,\gamma}^4\right] + \E\left[\left|\sum_{i=2}^{n}W_{n,i,\gamma}^2 - 1\right|^2\right]
\end{align}
To complete the proof, we proceed by deriving upper bounds for the terms on the right-hand side of~\eqref{eq:bdd_N_alpha}. A key step in this analysis is the following identity, which holds under the null hypothesis for all \(1 \leq i \neq j \leq n\):
\begin{align}\label{eq:HbarK}
  \sfH_n(\bZ_i,\bZ_j) = \bsfK_n(X_i, X_j) - \bsfK_n(X_i, Y_j) - \bsfK_n(Y_i, X_j) + \bsfK_n(Y_i, Y_j).
\end{align}
This expression will serve as a crucial tool in establishing the bounds from Theorem~\ref{thm:berry_essen_H0}. In the following lemma we first establish the upper bound on the first term from \eqref{eq:bdd_N_alpha}.

\begin{lemma}\label{lemma:4thmmntbdd}
    Under the assumptions of Theorem~\ref{thm:berry_essen_H0},
    \small
    \begin{align}\label{eq:bdd_4th_moment}
    \sum_{i=2}^{n}\E\left[W_{n,i,\gamma}^4\right]\lesssim t_{n,\gamma,1}\frac{\E\left[\bsfK_n(X_1, X_2)^4\right]}{M_n^2} + t_{n,\gamma,2}\frac{\E\left[\bsfK_n(X_1, X_2)^2\bsfK_n(X_1, X_3)^2\right]}{M_n^2} + t_{n,\gamma,2}.
\end{align}
\normalsize
\end{lemma}
We postpone the proof Lemma~\ref{lemma:4thmmntbdd} to Section~\ref{sec:proof_of_4thmmntbdd}. In the following lemma we establish an upper bound on the second term in \eqref{eq:bdd_N_alpha}.

\begin{lemma}\label{lemma:2nd_term_bdd}
  Under the assumptions on Theorem~\ref{thm:berry_essen_H0},
  \begin{align*}
    \E\left[\left|\sum_{i=2}^{n}W_{n,i,\gamma}^2 - 1\right|^2\right]\lesssim R_{n,\gamma,1} + R_{n,\gamma,2},
  \end{align*}
  where $R_{n,\gamma,1}$ and $R_{n,\gamma,2}$ are defined in \eqref{eq:def_Rn1} and \eqref{eq:def_Rn2} respectively.
\end{lemma}
The proof of Lemma~\ref{lemma:2nd_term_bdd} is postponed to Section~\ref{sec:proofof_second_term}. The proof of Theorem~\ref{thm:berry_essen_H0} is now completed by recalling the bound on $N_{n,\gamma}$ from \eqref{eq:bdd_N_alpha} and substituting the upper bounds on the individual terms from Lemma~\ref{lemma:4thmmntbdd} and Lemma~\ref{lemma:2nd_term_bdd}.

\subsubsection{Proof of Lemma~\ref{lemma:4thmmntbdd}}\label{sec:proof_of_4thmmntbdd}
By definition, if $0\leq \gamma< 1$,
\small
\begin{align*}
    \sum_{i=2}^{n}\E\left[W_{n,i,\gamma}^4\right]
    & = \sum_{i=2}^{n}\frac{1}{n^{4(1-\gamma)}M_n^2i^{4\gamma}c_\gamma^2}\sum_{j_1,j_2,j_3,j_4=1}^{i-1}\E\left[\prod_{\ell=1}^{4}\sfH_n(\bZ_{i}, \bZ_{j_{\ell}})\right]\\
    & = \sum_{i=2}^{n}\frac{1}{n^{4(1-\gamma)}M_n^2i^{4\gamma}c_\gamma^2}\left[\sum_{j=1}^{i-1}\E\left[\sfH_n(\bZ_i, \bZ_j)^4\right] + \sum_{j_1\neq j_2 = 1}^{i-1}\E\left[\sfH(\bZ_i, \bZ_{j_1})^2\sfH(\bZ_i, \bZ_{j_2})^2\right]\right],
\end{align*}
\normalsize
where the last equality follows by recalling that under null $\E\left[\sfH(\bZ_1,\bZ_2)\mid \bZ_1\right] = 0$. Continuing with the chain of equalities
\small
\begin{align}\label{eq:4thmomentbdd}
    \sum_{i=2}^{n}\E\left[W_{n,i,\gamma}^4\right]
    & = \sum_{i=2}^{n}\frac{1}{n^{4(1-\gamma)}M_n^2i^{4\gamma}c_\gamma^2}\bigg[(i-1)\E\left[\sfH_n(\bZ_1, \bZ_2)^4\right] \nonumber\\
    &\hspace{130pt} + (i-1)(i-2)\E\left[\sfH(\bZ_1, \bZ_2)^2\sfH(\bZ_1, \bZ_3)^2\right]\bigg]\nonumber\\
    &\lesssim \frac{\E\left[\sfH_n(\bZ_1, \bZ_2)^4\right]}{M_n^2}\sum_{i=2}^{n}\frac{i^{1-4\gamma}}{n^{4(1-\gamma)}} + \frac{\E\left[\sfH(\bZ_1, \bZ_2)^2\sfH(\bZ_1, \bZ_3)^2\right]}{M_n^2}\sum_{i=2}^{n}\frac{i^{2-4\gamma}}{n^{4(1-\gamma)}}.
\end{align}
\normalsize
Now by Lemma~\ref{lemma:sumip} we get
\small
\begin{align*}
    \sum_{i=2}^{n}\frac{i^{1-4\gamma}}{n^{4(1-\gamma)}} = 
    \begin{cases}
        \Theta(n^{-2}) & \text{ if }0\leq \gamma<\frac{1}{2}\\
        \Theta\left(\frac{\log n}{n^2}\right) & \text{ if }\gamma = \frac{1}{2}\\
        \Theta\left(n^{-4(1-\gamma)}\right) & \text{ if }\frac{1}{2}<\gamma< 1.
    \end{cases}
    \text{ and }
    \sum_{i=2}^{n}\frac{i^{2-4\gamma}}{n^{4(1-\gamma)}} = 
    \begin{cases}
        \Theta(n^{-1}) & \text{ if }0\leq \gamma<\frac{3}{4}\\
        \Theta\left(\frac{\log n}{n}\right) & \text{ if }\gamma = \frac{3}{4}\\
        \Theta\left(n^{-4(1-\gamma)}\right) & \text{ if }\frac{3}{4}<\gamma< 1.
    \end{cases}.
\end{align*}
\normalsize
Hence recalling the bound from \eqref{eq:4thmomentbdd} and notation from \eqref{eq:tnalpha} we get
\begin{align}\label{eq:4thbddless1}
    \sum_{i=2}^{n}\E\left[W_{n,i,\gamma}^4\right] \lesssim t_{n,\gamma,1}\frac{\E\left[\sfH_n(\bZ_1, \bZ_2)^4\right]}{M_n^2} + t_{n,\gamma,2}\frac{\E\left[\sfH(\bZ_1, \bZ_2)^2\sfH(\bZ_1, \bZ_3)^2\right]}{ M_n^2}.
\end{align}
Furthermore, if $\gamma = 1$, then similar computations as above shows
\begin{align}\label{eq:4thbdd1}
    \sum_{i=2}^{n}\E\left[W_{n,i,\gamma}^4\right]  
    & \lesssim \frac{\E\left[\sfH_n(\bZ_1, \bZ_2)^4\right]}{M_n^2}\sum_{i=2}^{n}\frac{i^{-3}}{(\log n)^2} + \frac{\E\left[\sfH(\bZ_1, \bZ_2)^2\sfH(\bZ_1, \bZ_3)^2\right]}{M_n^2}\sum_{i=2}^{n}\frac{i^{-2}}{(\log n)^2}\nonumber\\
    & \lesssim\frac{\E\left[\sfH_n(\bZ_1, \bZ_2)^4\right]}{(\log n)^2M_n^2} + \frac{\E\left[\sfH(\bZ_1, \bZ_2)^2\sfH(\bZ_1, \bZ_3)^2\right]}{(\log n)^2 M_n^2}.
\end{align}
The proof is now completed by noting that under null, $\E\left[\sfH_n(\bZ_1, \bZ_2)^4\right]\lesssim\E\left[\bsfK_n(X_1,X_2)^4\right]$ and 
\begin{align}\label{eq:Hn22bddKn22}
    \E\left[\sfH_n(\bZ_1, \bZ_2)^2\sfH_n(\bZ_1, \bZ_3)^2\right]\lesssim \E\left[\bsfK_n(X_1, X_2)^2\bsfK_n(X_1, X_3)^2\right] + M_n^2.
\end{align}

\subsubsection{Proof of Lemma~\ref{lemma:2nd_term_bdd}}\label{sec:proofof_second_term}
By expanding the square,
\begin{align}\label{eq:quad_form_1}
    \E\left[\left|\sum_{i=2}^{n}W_{n,i,\gamma}^2 - 1\right|^2\right] = \E\left[\left(\sum_{i=2}^{n}W_{n,i,\gamma}^2\right)^2\right] - 2\sum_{i=2}^{n}\E\left[W_{n,i,\gamma}^2\right] + 1.
\end{align}
Note that if $0\leq\gamma<1$, then
\begin{align}\label{eq:expand2ndmmnt}
    \sum_{i=2}^{n}\E\left[W_{n,i,\gamma}^2\right] 
    & = \sum_{i=2}^{n}\frac{1}{n^{2(1-\gamma)}M_ni^{2\gamma}c_\gamma}\E\left[\sum_{j_1,j_2 = 1}^{i-1}\sfH_n(\bZ_i,\bZ_{j_1})\sfH_n(\bZ_i, \bZ_{j_2})\right]\nonumber\\
    & = \frac{\E\left[\sfH_n(\bZ_1,\bZ_2)^2\right]}{M_nc_\gamma}\frac{1}{n^{2(1-\gamma)}}\sum_{i=2}^{n}\frac{i-1}{i^{2\gamma}}.
\end{align}
Moreover, expanding the term $\E\left[\sfH_n(\bZ_1, \bZ_2)^2\right]$ note that
\begin{align}\label{eq:evalH2}
    \E\left[\sfH_n(\bZ_1, \bZ_2)^2\right] 
    & = 4\E\left[\bsfK(X_1, X_2)^2\right] - 8\E\left[\bsfK(X_1, X_2)\bsfK(X_1, X_3)\right] = 4\E\left[\bsfK(X_1, X_2)^2\right]
\end{align}
where the last equality follows by noticing that $\E\left[\bsfK_n(X_1, X_2)\mid X_1\right] = 0$.
Hence for $0\leq \gamma<1$ from \eqref{eq:expand2ndmmnt} and \eqref{eq:evalH2} we conclude
\begin{align*}
    \sum_{i=2}^{n}\E\left[W_{n,i,\gamma}^2\right] = \frac{4}{c_\gamma}\frac{1}{n^{2(1-\gamma)}}\sum_{i=2}^{n}\frac{i-1}{i^{2\gamma}}.
\end{align*}
Similar computations for $\gamma = 1$ shows
\begin{align*}
    \sum_{i=2}^{n}\E\left[W_{n,i,1}^2\right] = \frac{1}{c_1\log n}\sum_{i=2}^{n}\frac{i-1}{i^2}.
\end{align*}
Then substituting the values in \eqref{eq:quad_form_1} we get
\small
\begin{align}\label{eq:Wnialpha_1_sq_decomp}
    \E\left[\left|\sum_{i=2}^{n}W_{n,i,\gamma}^2 - 1\right|^2\right] = 
    \begin{cases}
        \E\left[\left(\sum_{i=2}^{n}W_{n,i,\gamma}^2\right)^2\right] - \frac{8}{c_\gamma}\frac{1}{n^{2(1-\gamma)}}\sum_{i=2}^{n}\frac{i-1}{i^{2\gamma}} + 1 & \text{ if }0\leq \gamma<1\\
        \E\left[\left(\sum_{i=2}^{n}W_{n,i,\gamma}^2\right)^2\right] - \frac{2}{c_1\log n}\sum_{i=2}^{n}\frac{i-1}{i^2} + 1 & \text{ if }\gamma = 1.
    \end{cases}
\end{align}
\normalsize
Expanding the first term shows
\begin{align}\label{eq:Wnialpha_1_sq_expand}
    \E\left[\left(\sum_{i=2}^{n}W_{n,i,\gamma}^2\right)^2\right] = \sum_{i=2}^{n}\E\left[W_{n,i,\gamma}^4\right] + 2\sum_{i_1<i_2}\E\left[W_{n,i_1,\gamma}^2W_{n,i_2,\gamma}^2\right].
\end{align}
In the above expansion, note that applying Lemma~\ref{lemma:4thmmntbdd}, the first term can be upper bounded by $R_{n,\gamma,1}$. Consequently in the following we further expand the second term $\sum_{i_1<i_2}\E\left[W_{n,i_1,\gamma}^2W_{n,i_2,\gamma}^2\right]$. To that end for $1\leq m\leq 4$ we begin by defining the set of ordered tuples,
\begin{align*}
    \cS_{i_1,i_2,m} = \{(k_1,k_2, \ell_1,\ell_2): \left|\{j_1, j_2, \ell_1, \ell_2\}\right| = m, 1\leq j_1, j_2\leq i_1-1, 1\leq \ell_1, \ell_2\leq i_2-1\},
\end{align*}
where $\cS_{i_1,i_2,m}$ collects all the ordered four-tuples such that the first two coordinated can take values at most $i_1-1$, the second two coordinates can take values at most $i_2-1$ and there are exactly $m$ many distinct values in the tuple. Finally for $1\leq m\leq 4$ define,
\begin{align}\label{eq:def_L_mnalpha}
    L_{m,n,\gamma} = 
    \begin{cases}
        \frac{1}{n^4(1-\gamma)M_n^2}\displaystyle{\sum_{i_1<i_2}}\frac{1}{i_1^{2\gamma}i_2^{2\gamma}}\sum_{\cS_{i_1,i_2, m}}\E\left[\prod_{t=1}^2\sfH_n(\bZ_{i_1}, \bZ_{j_t})\sfH_n(\bZ_{i_2}, \bZ_{\ell_t})\right]& \text{ if }0\leq \gamma<1\\
        \frac{1}{(M_n\log n)^2}\displaystyle{\sum_{i_1<i_2}}\frac{1}{i_1^{2}i_2^{2}}\sum_{\cS_{i_1,i_2, m}}\E\left[\prod_{t=1}^2\sfH_n(\bZ_{i_1}, \bZ_{j_t})\sfH_n(\bZ_{i_2}, \bZ_{\ell_t})\right]& \text{ if }\gamma=1.
    \end{cases}
\end{align}
Then, by definition we have the following expansion.
\begin{align}\label{eq:Wnproddecomp}
    \sum_{i_1< i_2 = 2}^{n}\E\left[W_{n,i_1,\gamma}^2W_{n,i_2,\gamma}^2\right] = \frac{1}{c_\gamma^2}\sum_{m=1}^4L_{m,n,\gamma} = \frac{1}{c_\gamma^2}\sum_{m=1}^2L_{m,n,\gamma},
\end{align}
where the final equality follows by noticing that under $\bH_0$, $\E\left[\sfH_n(\bZ_1,\bZ_2)\mid \bZ_1\right] = 0$ almost surely. Next, we further decompose the term $L_{2,n,\gamma}$. Define
\begin{align*}
    \cS_{i_1,i_2,2}^\circ = \left\{(j_1, j_2,\ell_1, \ell_2): j_1 = j_2 = j, \ell_1=\ell_2=\ell, 1\leq k\leq i_1-1, 1\leq \ell\leq i_2-1, k\neq \ell\right\},
\end{align*}
Then, by construction $\cS_{i_1,i_2,2}^\circ\subseteq\cS_{i_1,i_2,2}$. Let
\begin{align}\label{eq:def_L2_circ}
    L_{2,n,\gamma}^\circ = 
    \begin{cases}
        \frac{1}{n^4(1-\gamma)M_n^2}\displaystyle{\sum_{i_1<i_2}}\frac{1}{i_1^{2\gamma}i_2^{2\gamma}}\sum_{\cS_{i_1,i_2, 2}^\circ}\E\left[\prod_{t=1}^2\sfH_n(\bZ_{i_1}, \bZ_{j_t})\sfH_n(\bZ_{i_2}, \bZ_{\ell_t})\right]& \text{ if }0\leq \gamma<1\\
        \frac{1}{(M_n\log n)^2}\displaystyle{\sum_{i_1<i_2}}\frac{1}{i_1^{2}i_2^{2}}\sum_{\cS_{i_1,i_2, 2}^\circ}\E\left[\prod_{t=1}^2\sfH_n(\bZ_{i_1}, \bZ_{j_t})\sfH_n(\bZ_{i_2}, \bZ_{\ell_t})\right]& \text{ if }\gamma=1.
    \end{cases}
\end{align}
We now decompose $L_{2,n,\gamma}$ as $L_{2,n,\gamma} = L_{2,n,\gamma}^\circ + (L_{2,n,\gamma} - L_{2,n,\gamma}^\circ)$. With the notations in place, combining \eqref{eq:Wnialpha_1_sq_decomp}, \eqref{eq:Wnialpha_1_sq_expand} and \eqref{eq:Wnproddecomp} we have
\begin{align}\label{eq:Wnialpha_1_sq_terms}
    \E\left[\left|\sum_{i=2}^{n}W_{n,i,\gamma}^2 - 1\right|^2\right] = \sum_{i=2}^{n}\E\left[W_{n,i,\gamma}^4\right] + \frac{2}{c_\gamma^2}L_{1,n,\gamma} + \frac{2}{c_\gamma^2}(L_{2,n,\gamma} - L_{2,n,\gamma}^\circ) + \zeta_{n,\gamma},
\end{align}
where
\begin{align*}
    \zeta_{n,\gamma} = 
    \begin{cases}
        \frac{2}{c_\gamma^2}L_{2,n,\gamma}^\circ - \frac{8}{c_\gamma}\frac{1}{n^{2(1-\gamma)}}\sum_{i=2}^{n}\frac{i-1}{i^{2\gamma}} + 1 & \text{ if }0\leq \gamma<1\\
        \frac{2}{c_\gamma^2}L_{2,n,\gamma}^\circ - \frac{2}{c_1\log n}\sum_{i=2}^{n}\frac{i-1}{i^2} + 1 & \text{ if }\gamma = 1.
    \end{cases}
\end{align*}
In the following we complete the proof by establishing bound on the individual terms in \eqref{eq:Wnialpha_1_sq_terms}. Note that by Lemma~\ref{lemma:4thmmntbdd} we already know that $\sum_{i=2}^{n}\E\left[W_{n,i,\gamma}^4\right]\lesssim R_{n,\gamma,1}$. Next, we bound the term $L_{1,n,\gamma}$. 

\paragraph{Bound on $L_{1,n,\gamma}$:}

For $L_{1,n,\gamma}$, note that for any $i_1<i_2$ with $(j_1, j_2, \ell_1,\ell_2)\in \cS_{i_1,i_2,1}$,
\begin{align*}
    \E\left[\prod_{t=1}^2\sfH_n(\bZ_{i_1}, \bZ_{j_t})\sfH_n(\bZ_{i_2}, \bZ_{\ell_t})\right] = \E\left[\sfH_n(\bZ_1,\bZ_2)^2\sfH_n(\bZ_1,\bZ_3)^2\right],
\end{align*}
and $|\cS_{i_1, i_2, 1}| = i_1-1$. Then for $0\leq \gamma<1$,
\begin{align}\label{eq:L1alpha0}
    L_{1,n,\gamma} 
    & \lesssim \frac{\E\left[\sfH_n(\bZ_1,\bZ_2)^2\sfH_n(\bZ_1,\bZ_3)^2\right]}{M_n^2}\frac{1}{n^{4(1-\gamma)}}\sum_{i_1<i_2}\frac{i_1-1}{i_1^{2\gamma}i_2^{2\gamma}}\nonumber\\
    &\lesssim\frac{\E\left[\sfH_n(\bZ_1,\bZ_2)^2\sfH_n(\bZ_1,\bZ_3)^2\right]}{M_n^2}\frac{1}{n^{4(1-\gamma)}}\sum_{i_2 = 2}^{n}\frac{1}{i_2^{2\gamma}}\sum_{i_1=1}^{i_2-1}i_1^{1-2\gamma}\nonumber\\
    &\lesssim\frac{\E\left[\sfH_n(\bZ_1,\bZ_2)^2\sfH_n(\bZ_1,\bZ_3)^2\right]}{M_n^2}t_{n,\gamma,2},
\end{align}
where the final bound follows from Lemma~\ref{lemma:sumip}. Similarly for $\gamma = 1$ we get
\begin{align}\label{eq:L110}
    L_{1,n,\gamma} 
    & \lesssim \frac{\E\left[\sfH_n(\bZ_1,\bZ_2)^2\sfH_n(\bZ_1,\bZ_3)^2\right]}{M_n^2}\frac{1}{(\log n)^2}\sum_{i_2 = 2}^{n}\frac{1}{i_2^{2}}\sum_{i_1=1}^{i_2-1}i_1^{-1}\nonumber\\
    & \lesssim\frac{\E\left[\sfH_n(\bZ_1,\bZ_2)^2\sfH_n(\bZ_1,\bZ_3)^2\right]}{M_n^2}t_{n,1,2}.
\end{align}
Hence recalling \eqref{eq:Hn22bddKn22} and using the bounds \eqref{eq:L1alpha0} and \eqref{eq:L110} we conclude
\small
\begin{align}\label{eq:L_1_n_alpha_bdd}
    L_{1,n,\gamma}\lesssim \frac{\E\left[\sfH_n(\bZ_1,\bZ_2)^2\sfH_n(\bZ_1,\bZ_3)^2\right]}{M_n^2}t_{n,\gamma,2}\lesssim \frac{\E\left[\bsfK_n(X_1,X_2)^2\bsfK_n(X_1, X_3)^2\right]}{M_n^2}t_{n,\gamma,2} + t_{n,\gamma,2}.
\end{align}
\normalsize
In the following we bound the term $\zeta_{n,\gamma}$.

\paragraph{Bound on $\zeta_{n,\gamma}$:} To bound $\zeta_{n,\gamma}$ we first simplify the term $L_{2,n,\gamma}^\circ$. If $0\leq \gamma<1$ then by construction,
\begin{align}\label{eq:L2_circ_alpha_equal}
    L_{2,n,\gamma}^\circ 
    & = \frac{1}{n^{4(1-\gamma)}M_n^2}\sum_{i_1<i_2}\frac{1}{i_1^{2\gamma}i_2^{2\gamma}}\sum_{j=1}^{i_1-1}\sum_{\ell=1}^{i_2-1}\E\left[\sfH_n(\bZ_j,\bZ_\ell)^2\right]^2\one\left\{j\neq \ell\right\}\nonumber\\
    & = \frac{16}{n^{4(1-\gamma)}}\sum_{i_1<i_2}\frac{(i_1-1)(i_2-2)}{i_1^{2\gamma}i_2^{2\gamma}},
\end{align}
where the last equality follows from \eqref{eq:evalH2}. Similarly if $\gamma = 1$, then
\begin{align}\label{eq:L2_circ_1_equal}
    L_{2,n,\gamma}^\circ 
    & = \frac{16}{(\log n)^2}\sum_{i_1<i_2}\frac{(i_1-1)(i_2-2)}{i_1^{2}i_2^{2}}.
\end{align}
Now for $0\leq \gamma<1$ with the identity from \eqref{eq:L2_circ_alpha_equal} notice,
\begin{align*}
    \zeta_{n,\gamma} 
    & = \frac{2}{c_\gamma^2}L_{2,n,\gamma}^\circ - \frac{8}{c_\gamma}\frac{1}{n^{2(1-\gamma)}}\sum_{i=2}^{n}\frac{i-1}{i^{2\gamma}} + 1\\
    & = \left(\frac{4}{c_\gamma n^{2(1-\gamma)}}\sum_{i=2}^{n}\frac{i-1}{i^{2\gamma}} - 1\right)^2 - \frac{32}{c_\gamma^2n^{4(1-\gamma)}}\sum_{i<j}\frac{(i-1)}{i^{2\gamma}j^{2\gamma}} - \frac{16}{c_\gamma^2n^{4(1-\gamma)}}\sum_{i=2}^{n}\left(\frac{i-1}{i^{2\gamma}}\right)^2.
\end{align*}
By Lemma~\ref{lemma:sumip} we get
\begin{align*}
    \frac{1}{n^{4(1-\gamma)}}\sum_{i=2}^{n}\left(\frac{i-1}{i^{2\gamma}}\right)^2 \lesssim 
    \begin{cases}
        1/n & \text{ if }0\leq \gamma<3/4\\
        \log n/n & \text{ if }\gamma = 3/4\\
        1/n^{4(1-\gamma)} & \text{ if }3/4<\gamma<1.
    \end{cases}
\end{align*}
Similarly applying Lemma~\ref{lemma:sumip} repeatedly we find
\begin{align*}
    \frac{1}{n^{4(1-\gamma)}}\sum_{i<j}\frac{(i-1)}{i^{2\gamma}j^{2\gamma}}\lesssim
    \begin{cases}
        1/n & \text{ if }0\leq \gamma<3/4\\
        \log n/n & \text{ if }\gamma = 3/4\\
        1/n^{4(1-\gamma)} & \text{ if }3/4<\gamma<1.
    \end{cases}
\end{align*}
Furthermore, from Lemma~\ref{lemma:i_2alpha_convg_rate} and recalling definition of $c_\gamma$ from \eqref{eq:def_c_alpha} we know
\begin{align*}
        \left|\frac{4}{c_\gamma n^{2(1-\gamma)}}\sum_{i=2}^{n}\frac{i-1}{i^{2\gamma}} - 1\right|\lesssim
        \begin{cases}
            \frac{1}{n} & \text{ if }0\leq \gamma<1/2\\
            \frac{\log n}{n} & \text{ if }\gamma = 1/2\\
            \frac{1}{n^{2(1-\gamma)}} & \text{ if }1/2<\gamma<1
        \end{cases}.
\end{align*}
Combining we conclude
\begin{align*}
    \left|\frac{2}{c_\gamma^2}L_{2,n,\gamma}^\circ - \frac{8}{c_\gamma}\frac{1}{n^{2(1-\gamma)}}\sum_{i=2}^{n}\frac{i-1}{i^{2\gamma}} + 1\right|\lesssim t_{n,\gamma,2} = 
    \begin{cases}
        1/n & \text{ if }0\leq \gamma<3/4\\
        \log n/n & \text{ if }\gamma = 3/4\\
        1/n^{4(1-\gamma)} & \text{ if }3/4<\gamma<1
    \end{cases}.
\end{align*}
Similarly for $\gamma = 1$, using the fact $\sum_{i=1}^{n}\frac{1}{i} = \log n + O(1)$ (see Theorem 1 from \cite{apostol1999elementary}) and using the identity from \eqref{eq:L2_circ_1_equal} we get,
\begin{align*}
    \left|\frac{2}{c_\gamma^2}L_{2,n,\gamma}^\circ - \frac{8}{c_\gamma}\frac{1}{\log n}\sum_{i=2}^{n}\frac{i-1}{i^{2}} + 1\right|\lesssim\frac{1}{(\log n)^2} = t_{n,1,2}.
\end{align*}
Hence we conclude that for all $0\leq \gamma\leq 1$,
\begin{align}\label{eq:zeta_bdd}
    \zeta_{n,\gamma}\lesssim t_{n,\gamma,2}.
\end{align}
Finally in the following we bound the difference $L_{2,n,\gamma} - L_{2,n,\gamma}^\circ$.

\paragraph{Bound on $L_{2,n,\gamma} - L_{2,n,\gamma}^\circ$}
To bound the difference $L_{2,n,\gamma} - L_{2,n,\gamma}^\circ$ note that,
\small 
\begin{align*}
    & |L_{2,n,\gamma} - L_{2,n,\gamma}^\circ|\\
    & \lesssim 
    \begin{cases}
        \frac{1}{n^{4(1-\gamma)}M_n^2}\displaystyle{\sum_{i_1<i_2}}\frac{1}{i_1^{2\gamma}i_2^{2\gamma}}\displaystyle{\sum_{j\neq\ell=1}^{i_1-1}}\left|\E\left[\sfH_n(\bZ_1,\bZ_2)\sfH_n(\bZ_2,\bZ_3)\sfH_n(\bZ_3,\bZ_4)\sfH_n(\bZ_4,\bZ_1)\right]\right| & 0\leq \gamma<1\\
        \frac{1}{(\log n)^2M_n^2}\sum_{i_1<i_2}\frac{1}{i_1^{2}i_2^{2}}\displaystyle{\sum_{j\neq\ell=1}^{i_1-1}}\left|\E\left[\sfH_n(\bZ_1,\bZ_2)\sfH_n(\bZ_2,\bZ_3)\sfH_n(\bZ_3,\bZ_4)\sfH_n(\bZ_4,\bZ_1)\right]\right| & \gamma = 1.
    \end{cases}
\end{align*}
\normalsize
The inner expectation evaluates to
\small
\begin{align}\label{eq:H4Kbar4}
    \E\bigg[\sfH(\bZ_1,\bZ_2)
    &\sfH(\bZ_2,\bZ_3)\sfH(\bZ_3,\bZ_4)\sfH(\bZ_4,\bZ_1)\bigg] = 16\E\left[\E\left[\bsfK(X_1, X_2)\bsfK(X_1,X_3)\mid X_2, X_3\right]^2\right].
\end{align}
\normalsize
The proof of \eqref{eq:H4Kbar4} is postponed to Section~\ref{sec:proofofH4Kbar4}. Then,
\small
\begin{align*}
    |L_{2,n,\gamma} - L_{2,n,\gamma}^\circ|
    & \lesssim 
    \begin{cases}
        \frac{1}{n^{4(1-\gamma)}M_n^2}\displaystyle{\sum_{i_1<i_2}}\frac{i_1^2}{i_1^{2\gamma}i_2^{2\gamma}}\E\left[\E\left[\bsfK(X_1, X_2)\bsfK(X_1,X_3)\mid X_2, X_3\right]^2\right] & 0\leq \gamma<1\\
        \frac{1}{(\log n)^2M_n^2}\sum_{i_1<i_2}\frac{1}{i_2^{2}}\E\left[\E\left[\bsfK(X_1, X_2)\bsfK(X_1,X_3)\mid X_2, X_3\right]^2\right] & \gamma = 1.
    \end{cases}
\end{align*}
\normalsize
Applying Lemma~\ref{lemma:sumip} we can simplify the above bound as,
\begin{align}\label{eq:L2_diff_bdd}
    |L_{2,n,\gamma} - L_{2,n,\gamma}^\circ|
    &\lesssim
    \begin{cases}
        \frac{\E\left[\E\left[\bsfK(X_1, X_2)\bsfK(X_1,X_3)\mid X_2, X_3\right]^2\right]}{M_n^2} & \text{ if }0\leq \gamma <1 \\
        \frac{\E\left[\E\left[\bsfK(X_1, X_2)\bsfK(X_1,X_3)\mid X_2, X_3\right]^2\right]}{M_n^2\log n} & \text{ if }\gamma = 1
    \end{cases}\nonumber\\
    &\lesssim
    \begin{cases}
        \frac{\E\left[\E\left[\bsfK(X_1, X_2)\bsfK(X_1,X_3)\mid X_2, X_3\right]^2\right]}{M_n^2} & \text{ if }0\leq \gamma <1\\
        \frac{1}{\log n} & \text{ if }\gamma = 1,
    \end{cases}
\end{align}
where the final bound follows from Cauchy-Schwartz inequality and recalling that $M_n = \E\left[\bsfK(X_1, X_2)^2\right]$. The proof of Lemma~\ref{lemma:2nd_term_bdd} is now completed by recalling the decomposition from \eqref{eq:Wnialpha_1_sq_terms} and combining the individual bounds from Lemma~\ref{lemma:4thmmntbdd}, \eqref{eq:L_1_n_alpha_bdd}, \eqref{eq:zeta_bdd} and \eqref{eq:L2_diff_bdd}.

\subsubsection{Proof of (\ref{eq:H4Kbar4}):}\label{sec:proofofH4Kbar4}
Let
\begin{align*}
    \bZ_i^{(0)} = X_i\text{ and }\bZ_i^{(1)} = Y_i\text{ for all }1\leq i\leq 4.
\end{align*}
Recalling \eqref{eq:HbarK} we get
\begin{align*}
    \sfH_n(\bZ_i,\bZ_j) = \sum_{\vep_i,\vep_j\in \{0,1\}}(-1)^{\vep_i+\vep_j}\bsfK_n\left(\bZ_i^{(\vep_i)},\bZ_j^{(\vep_j)}\right).
\end{align*}
Then,
\footnotesize
\begin{align*}
    &\sfH_n(\bZ_1,\bZ_2)\sfH_n(\bZ_2,\bZ_3)\sfH_n(\bZ_3,\bZ_4)\sfH_n(\bZ_4,\bZ_1)\\
    & = \sum_{\substack{\vep_1,\vep_2, \vep_3,\vep_4\\\vep_1^\prime,\vep_2^\prime, \vep_3^\prime,\vep_4^\prime}\in \{0,1\}}(-1)^{\sum_{\ell=1}^{4}\vep_\ell+\vep_\ell^\prime}\bsfK_n\left(\bZ_1^{(\vep_1)},\bZ_2^{(\vep_2)}\right)\bsfK_n\left(\bZ_2^{(\vep_2^\prime)},\bZ_3^{(\vep_3)}\right)\bsfK_n\left(\bZ_3^{(\vep_3^\prime)},\bZ_4^{(\vep_4)}\right)\bsfK_n\left(\bZ_4^{(\vep_4^\prime)},\bZ_1^{(\vep_1^\prime)}\right).
\end{align*}
\normalsize
Now recall that under null $\E\left[\bsfK_n\left(X_1,X_2\right)\mid X_1\right] = 0$. Hence it follows that
\begin{align*}
    &\E\bigg[\sfH_n(\bZ_1,\bZ_2)\sfH_n(\bZ_2,\bZ_3)\sfH_n(\bZ_3,\bZ_4)\sfH_n(\bZ_4,\bZ_1)\bigg]\\
    & = \sum_{\vep_1,\vep_2,\vep_3,\vep_4\in \{0,1\}}\E\left[\bsfK_n\left(\bZ_1^{(\vep_1)},\bZ_2^{(\vep_2)}\right)\bsfK_n\left(\bZ_2^{(\vep_2)},\bZ_3^{(\vep_3)}\right)\bsfK_n\left(\bZ_3^{(\vep_3)},\bZ_4^{(\vep_4)}\right)\bsfK_n\left(\bZ_4^{(\vep_4)},\bZ_1^{(\vep_1^\prime)}\right)\right]\\
    & = 16\E\left[\bsfK_n(X_1,X_2)\bsfK_n(X_2, X_3)\bsfK_n(X_3, X_4)\bsfK_n(X_4,X_1)\right]\\
    & = 16\E\left[\E\left[\bsfK_n(X_1, X_2)\bsfK_n(X_1,X_3)\mid X_2, X_3\right]^2\right].
\end{align*}

\subsection{Proof of Theorem~\ref{thm:general_consistency}}\label{sec:proofofgeneralconsistentcy}
For notational convenience define $s_n = \sum_{i=2}^{n}(i-1)/i^\gamma$ and
\begin{align*}
    R_{n,\gamma} = \frac{1}{s_n}\sum_{i=2}^{n}\frac{1}{i^\gamma}\sum_{j=1}^{i-1}\sfH(\bZ_i, \bZ_j),\  \hat{\sigma}_{n,\gamma}^2 = \frac{1}{s_n^2}\sum_{i=2}^{n}\left(\frac{1}{i^\gamma}\sum_{j=1}^{i-1}\sfH(\bZ_i, \bZ_j)\right)^2.
\end{align*}
Now we begin by noting that,
\begin{align*}
    \E\left[1-\phi_{n,\gamma}\right] = \P\left(R_{n,\gamma}\leq z_{\delta}\hat\sigma_{n,\gamma}\right).
\end{align*}
Fix $\vep>0$ and consider the event $\cE = \left\{\hat\sigma_{n,\gamma}^2\leq \E\left[\hat\sigma_{n,\gamma}^2\right]/\vep\right\}$. A simple application of Markov inequality shows that $\P(\cE^c)\leq \vep$. Hence, by decomposing the event $\{R_{n,\gamma}\leq z_\delta\hat\sigma_{n,\gamma}\}$ into $\cE$ and $\cE^c$ we get,
\begin{align*}
    \P\left(R_{n,\gamma}\leq z_{\delta}\hat\sigma_{n,\gamma}\right)\leq \P\left(R_{n,\gamma}\leq z_{\delta}\sqrt{\E\left[\hat\sigma_{n,\gamma}^2\right]/\vep}\right) + \vep.
\end{align*}
Now using the assumption \eqref{eq:genconsistent_assumption} from we know that for large enough $n$,
\begin{align*}
    z_{\delta}\sqrt{\E\left[\hat\sigma_{n,\gamma}^2\right]/\vep}\leq \frac{\delta_n^2}{2}.
\end{align*}
Furthermore by definition,
\begin{align*}
    \E\left[R_{n,\gamma}\right] = \MMD^2[P_n, Q_n, \sfK_n] = \delta_n^2.
\end{align*}
Combining these observations, for large enough $n$ we now get
\begin{align*}
    \P\left(R_{n,\gamma}\leq z_{\delta}\sqrt{\E\left[\hat\sigma_{n,\gamma}^2\right]/\vep}\right)
    &\leq \P\left(R_{n,\gamma} - \E\left[R_{n,\gamma}\right]\leq \delta_n^2/2 -\delta_n^2\right)\leq 4\frac{\Var\left(R_{n,\gamma}\right)}{\delta_n^4},
\end{align*}
where the final inequality follows from Chebyshev's inequality. This implies that
\begin{align*}
    \sup_{(P_n, Q_n)\in \cP_n} \P\left(R_{n,\gamma}\leq z_{\delta}\hat\sigma_{n,\gamma}\right)\leq \sup_{(P_n, Q_n)\in \cP_n}4\frac{\Var\left(R_{n,\gamma}\right)}{\delta_n^4} + \vep.
\end{align*}
The proof is now completed by recalling the assumption \eqref{eq:genconsistent_assumption} and then taking $n\ra\infty$ and $\vep\ra0$ respectively.

\subsection{Proof of Theorem~\ref{thm:optimality}}\label{sec:proof_minimax}
We divide the proof in two parts. In Section~\ref{sec:Type_I_minimax} we verify the type-I error control and in Section~\ref{sec:consistency_minimax} we verify the uniform consistency of the tests $\phi_{n,\gamma}$. 
\subsubsection{Type-I error: }\label{sec:Type_I_minimax}
To verify the Type-I error control we verify the conditions of Corollary~\ref{cor:NullCLT}. Recall that here the scale parameter is given by $\nu_n = n^\frac{4}{d+4\beta}$. Now recalling the expression of moments of Gaussian kernel derived by \citet{li2024optimality} we know
\begin{align}\label{eq:gausskernel_1}
    \E_{P_n}\left[\bsfK_n(X_1,X_2)^2\right] = \Omega(\nu_n^{-d/2}), \ \E_{P_n}\left[\bsfK_n(X_1,X_2)^4\right] \lesssim \nu_n^{-d/2},
\end{align}
and
\begin{align}\label{eq:gausskernel_2}
    \E_{P_n}\left[\bsfK_n(X_1,X_2)^2\bsfK_n(X_1,X_3)^2\right] \lesssim \nu_n^{-3d/4}.
\end{align}
By condition $(9)$ in \citet{li2024optimality} one can immediately verify that Assumption~\ref{assumption:kernelconvg} is satisfied. Now we need to verify that Assumption~\ref{assumption:4thand22moment} also holds true. We divide this into two parts. First note that
\begin{align*}
    t_{n,\gamma,1}\frac{\E_{P_n}\left[\bsfK_n(X_1, X_2)^4\right]}{\E_{P_n}\left[\bsfK_n(X_1, X_2)^2\right]^2}\lesssim t_{n,\gamma,1}\nu_n^{d/2} = t_{n,\gamma,1}n^{\frac{2d}{d + 4\beta}},
\end{align*}
where the upper bound follows from  \eqref{eq:gausskernel_1}. Now recalling the expression of $t_{n,\gamma,1}$ from \eqref{eq:tnalpha} we get
\begin{align*}
    t_{n,\gamma,1}n^{\frac{2d}{d + 4\beta}} = o(1) \text{ for all }0\leq \gamma < \frac{d + 8\beta}{2(d + 4\beta)}.
\end{align*}
Now for the second condition in \eqref{eq:tnalpha} note that
\begin{align*}
    t_{n,\gamma,2}\frac{\E_{P_n}\left[\bsfK_n(X_1, X_2)^2\bsfK_n(X_1, X_3)^2\right]}{\E_{P_n}\left[\bsfK_n(X_1, X_2)^2\right]^2}\lesssim t_{n,\gamma,2}\nu_n^{d/4} = t_{n,\gamma,2}n^{\frac{d}{d + 4\beta}}.
\end{align*}
where the upper bound now follows from \eqref{eq:gausskernel_2}. Now recalling the expression of $t_{n,\gamma,2}$ from \eqref{eq:tnalpha} we get
\begin{align*}
    t_{n,\gamma,2}n^{\frac{d}{d + 4\beta}} = o(1) \text{ for all }0\leq \gamma < \frac{3d+16\beta}{4(d + 4\beta)}.
\end{align*}
Combining we find that Assumption~\ref{assumption:4thand22moment} is satisfied whenever $0\leq \gamma <\frac{d + 8\beta}{2(d + 4\beta)}$. Hence, by applying Corollary~\ref{cor:NullCLT} we conclude that $\phi_{n,\gamma}$ has asymptotic type-I error control for $0\leq \gamma <\frac{d + 8\beta}{2(d + 4\beta)}$.

\subsubsection{Consistency}\label{sec:consistency_minimax} To verify consistency we need to show that conditions of Theorem~\ref{thm:general_consistency} are satisfied. Recalling \eqref{eq:genconsistent_assumption} we need to show that,
\small
\begin{align*}
    \lim_{n\ra\infty} \sup_{(P_n, Q_n)\in \cP_n} \frac{\E\left[\frac{1}{s_{n,\gamma}^2}\sum_{i=2}^n\left(\frac{1}{i^\gamma}\sum_{j=1}^{i-1}\sfH(\bZ_i,\bZ_j)\right)^2\right]}{\delta_n^4} + \frac{\Var\left(\frac{1}{s_{n,\gamma}}\sum_{i=2}^{n}\frac{1}{i^\gamma}\sum_{j=1}^{i-1}\sfH(\bZ_i, \bZ_j)\right)}{\delta_n^4} = 0
\end{align*}
\normalsize
where $s_{n,\gamma} = \sum_{i=2}^{n}\frac{i-1}{i^\gamma}$. We divide the proof into two parts to show separately,
\begin{align}\label{eq:consistency_1}
    \lim_{n\ra\infty} \sup_{(P_n, Q_n)\in \cP_n} \frac{\E\left[\frac{1}{s_{n,\gamma}^2}\sum_{i=2}^n\left(\frac{1}{i^\gamma}\sum_{j=1}^{i-1}\sfH(\bZ_i,\bZ_j)\right)^2\right]}{\delta_n^4} = 0
\end{align}
and,
\begin{align}\label{eq:consistency_2}
    \lim_{n\ra\infty} \sup_{(P_n, Q_n)\in \cP_n} \frac{\Var\left(\frac{1}{s_{n,\gamma}}\sum_{i=2}^{n}\frac{1}{i^\gamma}\sum_{j=1}^{i-1}\sfH(\bZ_i, \bZ_j)\right)}{\delta_n^4} = 0.
\end{align}
We begin with the proof of \eqref{eq:consistency_1}.
\paragraph{Proof of (\ref{eq:consistency_1}): }
We begin by expanding the numerator from \eqref{eq:consistency_1}.
\small
\begin{align}\label{eq:con1_num_expand}
    \E
    &\left[\frac{1}{s_{n,\gamma}^2}\sum_{i=2}^n\left(\frac{1}{i^\gamma}\sum_{j=1}^{i-1}\sfH(\bZ_i,\bZ_j)\right)^2\right]\nonumber\\
    & = \frac{1}{s_{n,\gamma}^2}\sum_{i=2}^{n}\frac{1}{i^{2\gamma}}\sum_{j=1}^{i-1}\E\left[\sfH_n^2(\bZ_i, \bZ_j)\right] + \frac{1}{s_{n,\gamma}^2}\sum_{i=2}^{n}\frac{1}{i^{2\gamma}}\sum_{j_1\neq j_2}\E\left[\sfH_n(\bZ_i, \bZ_{j_1})\sfH_n(\bZ_i, \bZ_{j_2})\right]\nonumber\\
    & = \frac{1}{s_{n,\gamma}^2}\sum_{i=2}^{n}\frac{i-1}{i^{2\gamma}}\E\left[\sfH_n(\bZ_1, \bZ_2)^2\right] + \frac{1}{s_{n,\gamma}^2}\sum_{i=2}^{n}\frac{(i-1)(i-2)}{i^{2\gamma}}\E\left[\sfH_n(\bZ_1,\bZ_2)\sfH_n(\bZ_1,\bZ_3)\right].
\end{align}
\normalsize
Now by definition,
\begin{align*}
    \E\left[\sfH_n(\bZ_1, \bZ_2)^2\right]\lesssim\E\left[\sfK_n(X_1, X_2)^2\right] + 2\E\left[\sfK_n(X_1, Y_2)^2\right] + \E\left[\sfK_n(Y_1, Y_2)^2\right].
\end{align*}
Recall that $\sfK_n$ is a Gaussian kernel and $\|p\|, \|q\|\leq M$. Hence applying Cauchy-Schwarz inequality on Lemma~\ref{lemma:gauss_kernel_prod} we get,
\begin{align*}
    \E\left[\sfH_n(\bZ_1, \bZ_2)^2\right]\lesssim \nu_n^{-d/2}.
\end{align*}
Furthermore,
\begin{align*}
    \E
    &\left[\sfH_n(\bZ_1,\bZ_2)\sfH_n(\bZ_1,\bZ_3)\right] = \E\left[\E\left[\sfH_n(\bZ_1,\bZ_2)\mid \bZ_1\right]^2\right]\\
    & \lesssim \E\left[\left(\E\left[\sfK_n(X_1, X_2)\mid X_1\right] - \E\left[\sfK_n(X_1, Y_2)\mid X_1\right]\right)^2\right]\\
    &\hspace{100pt} + \E\left[\left(\E\left[\sfK_n(X_2, Y_1)\mid Y_1\right] - \E\left[\sfK_n(Y_1, Y_2)\mid Y_1\right]\right)^2\right]\\
    & = \E\left[f(X)^2\right] + \E\left[f(Y)^2\right],
\end{align*}
where $f(\cdot) = \E\left[\sfK_n(X, \cdot)\right] - \E\left[\sfK_n(Y, \cdot)\right]$ with $X\sim p$ and $Y\sim q$. Then quoting Lemma 20 from \cite{shekhar2022permutation} we get
\begin{align*}
    \E&\left[\sfH_n(\bZ_1,\bZ_2)\sfH_n(\bZ_1,\bZ_3)\right] \lesssim \left\|p-q\right\|_{L_2}^2\nu_n^{-3d/4}.
\end{align*}
Now from Lemma 15 in \cite{li2024optimality} there exists $C>0$ such that for $g = p-q\in \cW^{\beta,2}(M)$,
\begin{align*}
    \left\|g\right\|_{L_2}^2\lesssim\int \exp\left(-\frac{\|\omega\|^2}{4\nu_n}\right)\left|\cF g(\omega)\right|^2\d \omega,
\end{align*}
given that $\nu_n\geq C\|g\|_{L_2}^{-2/\beta}$. Now note that by Lemma~\ref{lemma:gauss_kernel_prod} and that $\delta_n = \MMD[P_n, Q_n, \sfK_n]$,
\begin{align*}
    \int \exp\left(-\frac{\|\omega\|^2}{4\nu_n}\right)\left|\cF g(\omega)\right|^2\d \omega = \left(\frac{\nu_n}{\pi}\right)^{d/2}\delta_n^2,
\end{align*}
implying $\|p-q\|_{L_2}^2\lesssim\nu_n^{d/2}\delta_n^2$. Then,
\begin{align}\label{eq:bdd_two_star_H}
    \E\left[\sfH_n(\bZ_1,\bZ_2)\sfH_n(\bZ_1,\bZ_3)\right] \lesssim \left\|p-q\right\|_{L_2}^2\nu_n^{-3d/4}\lesssim\nu_n^{-d/4}\delta_n^2.
\end{align}
Substituting the bounds back in \eqref{eq:con1_num_expand}
\begin{align}\label{eq:bdd_t1_alt}
    \E\left[\frac{1}{s_{n,\gamma}^2}\sum_{i=2}^n\left(\frac{1}{i^\gamma}\sum_{j=1}^{i-1}\sfH(\bZ_i,\bZ_j)\right)^2\right]
    & \lesssim\frac{1}{s_{n,\gamma}^2}\sum_{i=2}^{n}\frac{i-1}{i^{2\gamma}}\nu_n^{-d/2} + \frac{1}{s_{n,\gamma}^2}\sum_{i=2}^{n}\frac{(i-1)(i-2)}{i^{2\gamma}}\nu_n^{-d/4}\delta_n^2\nonumber\\
    & \lesssim\frac{1}{n^2}\nu_n^{-d/2} + \frac{1}{n}\nu_n^{-d/4}\delta_n^2,
\end{align}
where the final bound follows from Lemma~\ref{lemma:sumip}. Then,
\small
\begin{align}\label{eq:bdd_2nd_moment}
    \frac{\E\left[\frac{1}{s_{n,\gamma}^2}\sum_{i=2}^n\left(\frac{1}{i^\gamma}\sum_{j=1}^{i-1}\sfH(\bZ_i,\bZ_j)\right)^2\right]}{\delta_n^4}\lesssim\frac{\nu_n^{-d/2}}{n^2\delta_n^4} + \frac{\nu_n^{-d/4}}{n\delta_n^2}\lesssim\frac{1}{\left(n^{2\beta/(d+4\beta)}\Delta_n\right)^4} + \frac{1}{\left(n^{2\beta/(d+4\beta)}\Delta_n\right)^2},
\end{align}
\normalsize
where the final bound follows by recalling that in $\cP_n^{(1)}$, $\Delta_n\leq \|p-q\|_2\leq \nu_n^{d/4}\delta_n$. The proof is now completed by recalling that $n^{2\beta/(d+4\beta)}\Delta_n\ra\infty$.

\paragraph{Proof of (\ref{eq:consistency_2}):}
Once again we begin by expanding the numerator from \eqref{eq:consistency_2}.
\small
\begin{align*}
    & \Var\left(\frac{1}{s_{n,\gamma}}\sum_{i=2}^{n}\frac{1}{i^\gamma}\sum_{j=1}^{i-1}\sfH(\bZ_i, \bZ_j)\right)\\
    & = \E\left[\frac{1}{s_{n,\gamma}^2}\sum_{i=2}^n\left(\frac{1}{i^\gamma}\sum_{j=1}^{i-1}\sfH_n(\bZ_i,\bZ_j)\right)^2\right] + 2\E\left[\frac{1}{s_{n,\gamma}^2}\sum_{i< j}\frac{1}{i^\gamma j^\gamma}\sum_{a=1}^{i-1}\sum_{b=1}^{j-1}\E\left[\sfH_n(\bZ_i, \bZ_a)\sfH_{n}(\bZ_j, \bZ_b)\right]\right] - \delta_n^4.
\end{align*}
\normalsize
Notice that we already know the bound on the first term from \eqref{eq:bdd_t1_alt}. In the following we analyse the second term only.
\small
\begin{align*}
    & \E\left[\frac{1}{s_{n,\gamma}^2}\sum_{i< j}\frac{1}{i^\gamma j^\gamma}\sum_{a=1}^{i-1}\sum_{b=1}^{j-1}\E\left[\sfH_n(\bZ_i, \bZ_a)\sfH_{n}(\bZ_j, \bZ_b)\right]\right]\\
    & = \frac{1}{s_{n,\gamma}^2}\sum_{i< j}\frac{1}{i^\gamma j^\gamma}\left[\sum_{a=1}^{i-1}\sum_{b=1}^{j-1}\E\left[\sfH_n(\bZ_i, \bZ_a)\sfH_{n}(\bZ_j, \bZ_b)\right]\one\left\{a\neq b\right\} + \sum_{a=1}^{i-1}\E\left[\sfH_n(\bZ_i, \bZ_a)\sfH_{n}(\bZ_j, \bZ_a)\right]\right]\\
    & = \frac{1}{s_{n,\gamma}^2}\sum_{i< j}\frac{(i-1)(j-2)}{i^\gamma j^\gamma}\delta_n^4 + \frac{1}{s_{n,\gamma}^2}\sum_{i<j}\frac{i-1}{i^\gamma j^\gamma}\E\left[\sfH_n(\bZ_1, \bZ_2)\sfH_n(\bZ_1, \bZ_3)\right].
\end{align*}
\normalsize
Now by Lemma~\ref{lemma:var_lim_half} we know that,
\footnotesize
\begin{align}\label{eq:bdd_var_2}
    &\lim_{n\ra\infty}\sup_{(P_n, Q_n)\in \cP_n^{(1)}}
    \frac{\Var\left(\frac{1}{s_{n,\gamma}}\sum_{i=2}^{n}\frac{1}{i^\gamma}\sum_{j=1}^{i-1}\sfH(\bZ_i, \bZ_j)\right)}{\delta_n^4}\nonumber\\
    & = \lim_{n\ra\infty}\sup_{(P_n, Q_n)\in \cP_n^{(1)}}\frac{\E\left[\frac{1}{s_{n,\gamma}^2}\sum_{i=2}^n\left(\frac{1}{i^\gamma}\sum_{j=1}^{i-1}\sfH(\bZ_i,\bZ_j)\right)^2\right] + \frac{2}{s_{n,\gamma}^2}\sum_{i<j}\frac{i-1}{i^\gamma j^\gamma}\E\left[\sfH_n(\bZ_1, \bZ_2)\sfH_n(\bZ_1, \bZ_3)\right]}{\delta_n^4}.
\end{align}
\normalsize
Now to complete the proof note that by Lemma~\ref{lemma:sumip},
\begin{align*}
    \frac{2}{s_{n,\gamma}^2}\sum_{i<j}\frac{i-1}{i^\gamma j^\gamma} \lesssim \frac{1}{n}.
\end{align*}
Hence continuing the trail of equalities from \eqref{eq:bdd_var_2} and recalling bounds from \eqref{eq:bdd_2nd_moment} and \eqref{eq:bdd_two_star_H} we get,
\begin{align*}
    \lim_{n\ra\infty}
    &\sup_{(P_n, Q_n)\in \cP_n^{(1)}}
    \frac{\Var\left(\frac{1}{s_{n,\gamma}}\sum_{i=2}^{n}\frac{1}{i^\gamma}\sum_{j=1}^{i-1}\sfH(\bZ_i, \bZ_j)\right)}{\delta_n^4}\\
    & \lesssim \frac{1}{\left(n^{2\beta/(d+4\beta)}\Delta_n\right)^4} + \frac{1}{\left(n^{2\beta/(d+4\beta)}\Delta_n\right)^2} + \frac{\nu_n^{-d/4}}{n\delta_n^2}\ra 0,
\end{align*}
which completes the proof.


\section{Proof of Theorem~\ref{thm:fixed_consistency} and Theorem~\ref{thm:general_alt_dist}}

\subsection{Proof of Theorem~\ref{thm:fixed_consistency}}\label{sec:proofof_fixed_consistency}
To prove Theorem~\ref{thm:fixed_consistency} we need only verify the conditions of Theorem~\ref{thm:gen_consistency} in this setting. Since $P,Q\in \cM_{\sfK}^{1}$ by a simple application of Cauchy-Schwartz ineqaulity all expectations in the following exists and are finite. Note that in this setting $\delta_n = \delta = \MMD[P,Q,\sfK]>0$. Then following computations similar to \eqref{eq:con1_num_expand} we get,
\begin{align}\label{eq:sigma_n_exp_0}
    \E\left[\sigma_n^2\right]\lesssim \frac{1}{n^2}\sum_{i=1}^{n}\frac{i-1}{i^2} + \frac{1}{n^2}\sum_{i=2}^{n}\frac{(i-1)(i-2)}{i^2} = O(1/n)\ra 0.
\end{align}
Moreover, similar to \eqref{eq:bdd_var_2} we get the following decompositions.
\small
\begin{align*}
    \Var(T_n) = \E\left[\frac{1}{n^2}\sum_{i=2}^n\left(\frac{1}{i}\sum_{j=1}^{i-1}\sfH(\bZ_i,\bZ_j)\right)^2\right] + 2\E\left[\frac{1}{n^2}\sum_{i< j}\frac{1}{ij}\sum_{a=1}^{i-1}\sum_{b=1}^{j-1}\E\left[\sfH_n(\bZ_i, \bZ_a)\sfH_{n}(\bZ_j, \bZ_b)\right]\right] - \delta^4.
\end{align*}
\normalsize
Furthermore,
\begin{align*}
    \E
    &\left[\frac{1}{n^2}\sum_{i< j}\frac{1}{ij}\sum_{a=1}^{i-1}\sum_{b=1}^{j-1}\E\left[\sfH(\bZ_i, \bZ_a)\sfH(\bZ_j, \bZ_b)\right]\right]\\
    & = \frac{1}{n^2}\sum_{i< j}\frac{(i-1)(j-2)}{i j}\delta^4 + \frac{1}{n^2}\sum_{i<j}\frac{i-1}{ij}\E\left[\sfH(\bZ_1, \bZ_2)\sfH(\bZ_1, \bZ_3)\right]\\
    & = \delta^4/2(1 + o(1)).
\end{align*}
Combined with \eqref{eq:sigma_n_exp_0} we can now easily conclude that $\Var(T_n)\ra 0$. This completes checking the conditions of Theorem~\ref{thm:gen_consistency} and subsequenlty proves Theorem~\ref{thm:fixed_consistency}.

\subsection{Proof of Theorem~\ref{thm:general_alt_dist}}\label{sec:proofof_thm_alt}
The statistic $T_n$ from \eqref{eq:Tn_alt} can be rewritten as $T_n = \frac{1}{2n}\sum_{i\neq j}w_{ij}\sfH_n(\bZ_i,\bZ_j)$ where $\sfH_n$ is defined in \eqref{eq:def_Hn_main} $w_{ij} = 1/\max\{i,j\}, 1\leq i\neq j\leq n$ and $\bZ_i = (X_i, Y_i), 1\leq i\leq n$. For notational simplicity take $\mu_n = \MMD[P_n, Q_n, \sfK_n]^2$ and write,
\begin{align*}
    \sfH_n(\bz_1, \bz_2) = \mu_n + h_n(\bz_1) + h_n(\bz_2) + g_n(\bz_1,\bz_2),
\end{align*}
where $g_n(\bz_1,\bz_2) = \sfH_n(\bz_1,\bz_2) - h_n(\bz_1) - h_n(\bz_2) + \mu_n$. Substituing back this expression in $T_n$ we get,
\begin{align*}
    \frac{T_n}{\mu_n} = 1 + \sum_{i=1}^{n}\underbrace{w_i^{(n)}h_n(\bZ_i)}_{W_{i,n}} + \frac{1}{2n\mu_n}\sum_{i\neq j}w_{ij}g(\bZ_i, \bZ_j) + O(\log n/n),
\end{align*}
where $w_i^{(n)} = \frac{i-1}{n\mu_ni} + \frac{1}{n\mu_n}\sum_{j=i+1}\frac{1}{j}$. Define $L_n = \sum_{i=1}^{n}W_{i,n}$. Moreover we can show,
\begin{align}\label{eq:w_insq_sum_lim}
    S_n := n\mu_n^2\sum_{i=1}^{n}\left(w_{i}^{(n)}\right)^2 \ra 5.
\end{align}
We postpone the proof of \eqref{eq:w_insq_sum_lim} to Section~\ref{sec:proof_of_S_n_convvg}. Now note that $w_{i}^{(n)}\lesssim H_n/n\mu_n$ for all $1\leq i\leq n$ and hence $\sum_{i=1}^{n}\left(w_i^{(n)}\right)^3\lesssim \frac{\left(\log n\right)^3}{n^3\mu_n^3}$. Now to prove CLT, we will invoke Lyapunov's version on $L_n$. To that end note that $\E\left[h_n(\bZ_1)\right] = 0$ which implies, $\E\left[W_{n,i}\right] = 0$ for all $1\leq i\leq n$ and,
\begin{align*}
    \frac{\sum_{i=1}^{n}\E\left[\left|W_{n,i}\right|^3\right]}{\left(\sum_{i=1}^{n}\E\left[W_{n,i}^2\right]\right)^{3/2}} = \frac{\E\left[\left|h_{n}(\bZ_1)\right|^3\right]}{\E\left[h_{n}(\bZ_1)^2\right]^{3/2}}\frac{\sum_{i=1}^{n}\left(w_i^{(n)}\right)^3}{\left[\sum_{i=1}^{n}\left(w_i^{(n)}\right)^2\right]^{3/2}} = O\left(\frac{(\log n)^3\E\left[\left|h_{n}(\bZ_1)\right|^3\right]}{n^{3/2}\Var\left(h_n(\bZ_1)\right)^{3/2}}\right).
\end{align*}
Hence by Lyapunov CLT we have $\frac{L_n}{\sqrt{\Var(L_n)}}\dto N(0,1)$. Now recalling the definition of $g_{n}$ note that $\E\left[g_n(\bZ_1,\bZ_2)\right] =0$ and
\begin{align*}
    \Var\left(\frac{1}{2n\mu_n}\sum_{i\neq j}w_{ij}g_n(\bZ_i,\bZ_j)\right) \lesssim\frac{\E\left[g_n(\bZ_1,\bZ_2)^2\right]}{\mu_n^2}\frac{H_n}{n^2}.
\end{align*}
Hence we get the bound,
\small
\begin{align*}
    \E\left[\frac{\left(\frac{1}{2n\mu_n}\sum_{i\neq j}w_{ij}g_n(\bZ_i,\bZ_j)\right)^2}{\Var(L_n)}\right]\lesssim \frac{\E\left[g_n(\bZ_1,\bZ_2)^2\right]}{\mu_n^2}\frac{H_n}{n^2} \frac{n\mu_n^2}{\E\left[h_{n}(\bZ_1)^2\right]} 
    & = O\left(\frac{\log n\E\left[g_n(\bZ_1,\bZ_2)^2\right]}{n\Var\left(h_{n}(\bZ_1)\right)}\right).
\end{align*}
\normalsize
Moreover,
\begin{align*}
    \frac{\log n}{n\sqrt{\Var(L_n)}} \lesssim \frac{\mu_n\log n}{\sqrt{n\Var\left(h_{n}(\bZ_1)^2\right)}}.
\end{align*}
Combining all the above and recalling the assumptions of Theorem~\ref{thm:general_alt_dist} we conclude,
\begin{align*}
    \frac{1}{\sqrt{\Var(L_n)}}\left(\frac{T_{n}}{\mu_n}-1\right)\dto \rm N(0,1).
\end{align*}
Moreover using \eqref{eq:w_insq_sum_lim} we know $n\mu_n^2\frac{\Var(L_n)}{\Var\left(h_{n}(\bZ_1)^2\right)}\ra 5$. Hence,
\begin{align*}
    \sqrt{\frac{n\mu_n^2}{\Var\left(h_{n}(\bZ_1)\right)}}\left(\frac{T_{n}}{\mu_n}-1\right)\dto \mathrm{N}(0,5)\implies \sqrt{\frac{n}{\Var\left(h_{n}(\bZ_1)\right)}}\left(T_{n}-\mu_n\right)\dto \rm N(0,5).
\end{align*}

\subsubsection{Proof of (\ref{eq:w_insq_sum_lim}):}\label{sec:proof_of_S_n_convvg}
By a simple computation,
\begin{align*}
    S_n = n\mu_n^2\sum_{i=1}^{n}\left(w_{i}^{(n)}\right)^2 = \frac{1}{n}\left[\sum_{i=1}^{n}\left(\frac{i-1}{i}\right)^2 + 2\sum_{i=1}^{n}\frac{i-1}{i}\sum_{j=i+1}^n\frac{1}{j} + \sum_{i=1}^{n}\left(\sum_{j=i+1}^{n}\frac{1}{j}\right)^2\right].
\end{align*}
For the last term note that
\begin{align*}
    \frac{1}{n}\sum_{i=1}^{n}\left(\sum_{j=i+1}^{n}\frac{1}{j}\right)^2 = \frac{1}{n}\sum_{j_1,j_2}\frac{\min\{j_1,j_2\}-1}{j_1j_2} = \frac{1}{n}\sum_{j_1,j_2}\frac{\min\{j_1,j_2\}}{j_1j_2} - \frac{H_n^2}{n},
\end{align*}
where $H_n$ is the harmonic sum till $n$. Now note that,
\begin{align*}
    \frac{1}{n}\sum_{j_1,j_2}\frac{\min\{j_1,j_2\}}{j_1j_2} = \frac{2}{n}\sum_{j_1\leq j_2}\frac{1}{j_2} - \frac{1}{n}\sum_{j}\frac{1}{j} = 2 - \frac{H_n}{n}.
\end{align*}
Hence,
\begin{align*}
    \frac{1}{n}\sum_{i=1}^{n}\left(\sum_{j=i+1}^{n}\frac{1}{j}\right)^2 = 2 - \frac{H_n}{n} + \frac{H_n^2}{n} \ra 2.
\end{align*}
For the first term we have,
\begin{align*}
    \frac{1}{n}\sum_{i=1}^{n}\left(\frac{i-1}{i}\right)^2 = 1 - 2\frac{H_n}{n} + \frac{1}{n}\sum_{i=1}^{n}\frac{1}{i^2} \ra 1.
\end{align*}
Finally for the second term,
\begin{align*}
    \frac{1}{n}\sum_{i=1}^{n}\frac{i-1}{i}\sum_{j=i+1}^n\frac{1}{j} = \frac{1}{n}\sum_{j=2}^{n}\frac{1}{j}\sum_{i=1}^{j-1}\frac{i-1}{i} = \frac{1}{n}\sum_{j=2}^{n}\frac{j-1-H_{j-1}}{j} = 1 - \frac{H_n}{n} - \frac{1}{n}\sum_{j=2}^{n}\frac{H_{j-1}}{j}\ra 1.
\end{align*}
Combining we can conclude that $S_n\ra 5$.


\section{Technical Results}

\begin{lemma}\label{lemma:E_barK_sq_positive}
  Let $\mathsf{K}$ be a kernel satisfying Assumption~\ref{assumption:K} and suppose that $X \sim P$ is not almost surely constant. Then, for independent samples $X_1, X_2 \overset{\text{i.i.d.}}{\sim} P$ with $P \in \mathcal{M}_{\mathsf{K}}^1(\mathcal{X})$, 
  \begin{align*}
    \E\left[\bsfK(X_1, X_2)^2\right]>0.
  \end{align*}
\end{lemma}

\begin{proof}
  Let $\cK\otimes \cK$ be the tensor product space (see \citet[Section 4.6]{berlinet2011reproducing}). Let $h(X_i) = \sfK(X_i,\cdot) - \nu_P, i=1,2$. With this notation we have,
  \begin{align*}
    \E\left[\bsfK(X_1, X_2)^2\right] = \E\left[\left\langle h(X_1)\otimes h(X_1),h(X_2)\otimes h(X_2)\right\rangle_{\cK\otimes \cK}\right] = \left\|\E\left[h(X)\otimes h(X)\right]\right\|_{\cK\otimes \cK}^2,
  \end{align*}
  where $X\sim P$. Now for the sake of contradiction assume that $\E\left[\bsfK(X_1, X_2)^2\right] = 0$. Then for any $r\in \cK$,
  \begin{align*}
    \E\left[\left\langle r, h(X)\right\rangle_{\cK}^2\right] = \E\left[\left\langle r\otimes r, h(X)\otimes h(X)\right\rangle\right] = 0.
  \end{align*}
  This implies that $\sfK(X_1,\cdot) = \sfK(X_2,\cdot)$ almost surely. Now recalling the characteristic property of $\sfK$ and considering the Dirac measures $\delta_x, x\in \cX$ we can immediately conclude that $x\to \sfK(x,\cdot)$ is an injective function. Hence $\sfK(X_1,\cdot) = \sfK(X_2,\cdot)$ almost surely implies that $X_1 = X_2$ almost surely, which leads to a contradiction since $X_1,X_2$ are generated independently from $P$.
\end{proof}

\begin{lemma}\label{lemma:sumip}
    For any $p$,
    \begin{align*}
        \sum_{i=1}^{n}i^p = 
        \begin{cases}
            \Theta\left(n^{p+1}\right) & \text{ if } p>-1\\
            \Theta\left(\log n\right) & \text{ if }p = -1\\
            \Theta(1) & \text{ if }p<-1
        \end{cases}.
    \end{align*}
\end{lemma}

\begin{proof}
We use the method of integral bounds by considering the function $f(x)=x^p$.

\paragraph{Case 1: $p > -1$}
Let $f(x) = x^p$.
If $p>0$, $f(x)$ is increasing, so
\[
    \int_1^n x^p \,dx \le \sum_{i=1}^n i^p \le \int_1^{n+1} x^p \,dx,
\]
which implies
\[
    \frac{n^{p+1}-1}{p+1} \le \sum_{i=1}^n i^p \le \frac{(n+1)^{p+1}-1}{p+1}.
\]
If $-1 < p < 0$, $f(x)$ is decreasing, so
\[
    \int_1^{n+1} x^p \,dx \le \sum_{i=1}^n i^p \le 1^p + \int_1^n x^p \,dx,
\]
which implies
\[
    \frac{(n+1)^{p+1}-1}{p+1} \le \sum_{i=1}^n i^p \le 1 + \frac{n^{p+1}-1}{p+1}.
\]
In both subcases, the lower and upper bounds are $\Theta(n^{p+1})$. Thus, $\sum_{i=1}^{n}i^p = \Theta(n^{p+1})$.

\paragraph{Case 2: $p = -1$}
The sum is the $n$-th harmonic number, $H_n = \sum_{i=1}^n \frac{1}{i}$. The function $f(x)=1/x$ is decreasing for $x>0$, so we apply the integral bounds:
\[
    \int_1^{n+1} \frac{1}{x} \,dx \le H_n \le 1 + \int_1^n \frac{1}{x} \,dx ~,
\]
\[
    \ln(n+1) \le H_n \le 1 + \ln n.
\]
Since $\ln(n+1)$ and $1+\ln n$ are both $\Theta(\ln n)$, it follows that $H_n = \Theta(\ln n)$.

\paragraph{Case 3: $p < -1$}
The sum's terms are positive, so it is bounded below by its first term:
\[
    \sum_{i=1}^n i^p \ge 1^p = 1 \implies \sum_{i=1}^n i^p = \Omega(1).
\]
Since $f(x)=x^p$ is decreasing, the sum is bounded above by the convergent integral:
\[
    \sum_{i=1}^n i^p \le 1 + \int_1^n x^p \,dx < 1 + \int_1^\infty x^p \,dx.
\]
We evaluate the integral, noting $p+1 < 0$:
\[
    \int_1^\infty x^p \,dx = \left[ \frac{x^{p+1}}{p+1} \right]_1^\infty = \lim_{b\to\infty}\frac{b^{p+1}}{p+1} - \frac{1}{p+1} = 0 - \frac{1}{p+1} = \frac{-1}{p+1}.
\]
The sum is bounded above by the constant $1 - \frac{1}{p+1}$, so $\sum_{i=1}^n i^p = O(1)$.
Combining the bounds gives $\sum_{i=1}^n i^p = \Theta(1)$.
\end{proof}

\begin{lemma}\label{lemma:i_2alpha_convg_rate}
    For every $n\geq 2$,
    \begin{align*}
        \left|\frac{1}{n^{2(1-\gamma)}}\sum_{i=2}^{n}\frac{i-1}{i^{2\gamma}} - \frac{1}{2(1-\gamma)}\right|\lesssim
        \begin{cases}
            \frac{1}{n} & \text{ if }0\leq \gamma<1/2\\
            \frac{\log n}{n} & \text{ if }\gamma = 1/2\\
            \frac{1}{n^{2(1-\gamma)}} & \text{ if }1/2<\gamma<1
        \end{cases}.
    \end{align*}
\end{lemma}

\begin{proof}
    Notice that
    \begin{align}\label{eq:i_1_2alpha_decomp}
        \frac{1}{n^{2(1-\gamma)}}\sum_{i=2}^{n}\frac{i-1}{i^{2\gamma}} = \frac{1}{n^{2(1-\gamma)}}\sum_{i=1}^{n}\frac{1}{i^{2\gamma-1}} - \frac{1}{n^{2(1-\gamma)}}\sum_{i=1}^{n}\frac{1}{i^{2\gamma}}.
    \end{align}
    Invoking Theorem 2 from \cite{apostol1999elementary} with $f(x) = x^{1-2\gamma}$ we get
    \begin{align*}
        \left|\sum_{i=1}^{n}\frac{1}{i^{2\gamma-1}} - \int_{1}^{n}x^{1-2\gamma}\d x\right|
        &\lesssim \int_{1}^{n}(1-2\gamma)x^{-2\gamma}\d x + 1\lesssim \frac{1}{n^{2\gamma-1}} + 1.
    \end{align*}
    Evaluating the integral gives
    \begin{align*}
        \left|\sum_{i=1}^{n}\frac{1}{i^{2\gamma-1}} - \frac{n^{2(1-\gamma)}}{2(1-\gamma)}\right|\lesssim \frac{1}{n^{2\gamma-1}} + 1 + \frac{1}{2(1-\gamma)}.
    \end{align*}
    Hence,
    \begin{align}\label{eq:i_2alpha_1_bdd}
        \left|\frac{1}{n^{2(1-\gamma)}}\sum_{i=1}^{n}\frac{1}{i^{2\gamma-1}} - \frac{1}{2(1-\gamma)}\right|\lesssim \frac{1}{n} + \frac{1}{n^{2(1-\gamma)}}.
    \end{align}
    Also, by Lemma~\ref{lemma:sumip} we know,
    \begin{align}\label{eq:i_2alpha_bdd}
        \frac{1}{n^{2(1-\gamma)}}\sum_{i=1}^{n}\frac{1}{i^{2\gamma}} = 
        \begin{cases}
            \Theta\left(1/n\right) & \text{ if }0\leq \gamma<1/2\\
            \Theta\left(\log n/n\right) & \text{ if }\gamma=1/2\\
            \Theta\left(1/n^{2(1-\gamma)}\right) & \text{ if }1/2<\gamma<1.
        \end{cases}
    \end{align}
    Now combining \eqref{eq:i_1_2alpha_decomp}, \eqref{eq:i_2alpha_1_bdd} and \eqref{eq:i_2alpha_bdd} we get,
    \begin{align*}
        \bigg|\frac{1}{n^{2(1-\gamma)}}
        &\sum_{i=2}^{n}\frac{i-1}{i^{2\gamma}} - \frac{1}{2(1-\gamma)}\bigg|\\
        &\lesssim  \left|\frac{1}{n^{2(1-\gamma)}}\sum_{i=1}^{n}\frac{1}{i^{2\gamma-1}} - \frac{1}{2(1-\gamma)}\right| +  \frac{1}{n^{2(1-\gamma)}}\sum_{i=1}^{n}\frac{1}{i^{2\gamma}}\lesssim
        \begin{cases}
            \frac{1}{n} & \text{ if }0\leq \gamma<1/2\\
            \frac{\log n}{n} & \text{ if }\gamma = 1/2\\
            \frac{1}{n^{2(1-\gamma)}} & \text{ if }1/2<\gamma<1
        \end{cases}
    \end{align*}
    which completes the proof.
\end{proof}

\begin{lemma}\label{lemma:var_lim_half}
Let $s_{n,\gamma} = \sum_{i=2}^{n}\frac{i-1}{i^\gamma}$. Then for $\gamma < 1$,
\begin{align*}
    \lim_{n\ra\infty}\frac{1}{s_{n,\gamma}^2}\sum_{i< j}\frac{(i-1)(j-2)}{i^\gamma j^\gamma} = \frac{1}{2}.
\end{align*}
\end{lemma}

\begin{proof}
We will analyze the denominator $s_{n,\gamma}$ and the numerator $\sum_{i< j}\frac{(i-1)(j-2)}{i^\gamma j^\gamma}$ separately. We begin with the denominator given by $s_{n,\gamma} = \sum_{i=2}^{n}\frac{i-1}{i^\gamma}$. We can split this sum:
$$s_{n,\gamma} = \sum_{i=2}^{n} \left( \frac{i}{i^\gamma} - \frac{1}{i^\gamma} \right) = \sum_{i=2}^{n} i^{1-\gamma} - \sum_{i=2}^{n} i^{-\gamma}.$$
For $\gamma < 1$ we consider the two terms separately. For the first term, once again applying the integral approximation from Theorem 2 of \cite{apostol1999elementary} we can easily show,
\begin{align*}
    \frac{1}{n^{2-\gamma}}\left|\sum_{i=2}^{n}i^{1-\gamma} - \int_{2}^{n}x^{1-\gamma}\d x\right| = o(1).
\end{align*}
Evaluating the inner integral we then conclude
\begin{align}\label{eq:i_1_alpha_lim}
    \frac{1}{n^{2-\gamma}}\sum_{i=2}^{n}i^{1-\gamma} - \frac{1}{2-\gamma} = o(1).
\end{align}
By a similar argument we can show $\frac{1}{n^{2-\gamma}}\sum_{i=2}^{n}i^{-\gamma} = o(1)$ and hence $s_{n,\gamma}/n^{2-\gamma} \ra 1/(2-\gamma)$.\\

To find the limit of the numerator, note that we can rewrite the numerator as
\begin{align*}
    N := \sum_{i<j}\frac{(i-1)(j-2)}{i^{\gamma}j^{\gamma}} = \frac{1}{2}\left[\left(\sum_{i=1}^{n}\frac{i-1}{i^\gamma}\right)^2 - \sum_{i=1}^{n}\left(\frac{i-1}{i}\right)^2\right] - \sum_{i<j}\frac{i-1}{i^\gamma j^\gamma}.
\end{align*}
Then applying Lemma~\ref{lemma:sumip} a simple computation yields
\begin{align*}
    \frac{N}{n^{4-2\gamma}} = \frac{1}{2}\left(\frac{1}{n^{2-\gamma}}\sum_{i=1}^{n}\frac{i-1}{i^\gamma}\right)^2 + o(1) = \frac{1}{2}\left(\frac{1}{n^{2-\gamma}}\sum_{i=1}^{n}i^{1-\gamma}\right)^2 + o(1).
\end{align*}
Recalling the limit from \eqref{eq:i_1_alpha_lim} we can conclude that $N/n^{2-4\gamma} \ra 1/2(2-\gamma)^2$. The proof is now completed by recalling the limit of $s_{n,\gamma}$.
\end{proof}

\begin{lemma}\label{lemma:gauss_kernel_prod}
    Consider $f, g\in L_2(\R^d)$. Then 
    \begin{align*}
        \int G_{\nu}(x,y)f(x)g(y)\d x\d y = \left(\frac{\pi}{\nu}\right)^{d/2}\int \exp\left(-\frac{\|\omega\|^2}{4\nu}\right)\overline{\cF f(\omega)}\cF g(\omega)\d\omega,
    \end{align*}
    where $G_\nu(x,y) = \exp(-\nu\left\|x-y\right\|_2^2)$ is the Gaussian kernel with bandwidth $\nu$, $\cF f(\cdot)$ and $\cF g(\cdot)$ are the corresponding Fourier transforms and $\overline{\cF f(\cdot)}$ is the complex conjugate.
\end{lemma}

\begin{proof}
    Consider $Z$ to be a Gaussian random vector with mean $\bm 0_d$ and covariance matrix $2\nu \bm I_d$. Then,
    \begin{align*}
        \int 
        & G_{\nu}(x,y)f(x)g(y)\d x\d y = \int \exp\left(-\nu\left\|x-y\right\|^2\right)f(x)g(y)\d x\d y\\
        & = \int \E \exp\left[\iota Z^\top(x-y)\right]f(x)g(y)\d x\d y\\
        & = \E\int\exp(\iota Z^\top x)f(x)\d x \int\exp(-\iota Z^\top y)g(y)\d y\\
        & = \int \frac{1}{(4\pi\nu)^{d/2}}\exp\left(-\frac{\|\omega\|^2}{4\nu}\right)\int\exp(\iota \omega^\top x)f(x)\d x \int\exp(-\iota \omega^\top y)g(y)\d y\d \omega\\
        & = \left(\frac{\pi}{\nu}\right)^{d/2}\int \exp\left(-\frac{\|\omega\|^2}{4\nu}\right)\overline{\cF f(\omega)}\cF g(\omega)\d\omega.
    \end{align*}
\end{proof}

\end{document}